\documentclass[12pt, draftclsnofoot, onecolumn]{IEEEtran}

\usepackage{algorithm,algorithmic,amsbsy,amsmath,amssymb,epsfig,mathrsfs,url,color,cite}
\usepackage{subcaption}
\usepackage{graphicx,epstopdf,tikz,pgfplots,amssymb,amsfonts,amsmath,mathrsfs,lettrine}
\usepackage{url}
\usepackage{enumerate}

\newcommand{\beq}{\begin{equation}}
\newcommand{\eeq}{\end{equation}}
\newcommand{\beqn}{\begin{eqnarray}}
\newcommand{\eeqn}{\end{eqnarray}}

\newtheorem{theorem}{\textbf{Theorem}}
\newtheorem{lem}{\textbf{Lemma}}

\newtheorem{col}{\textbf{Corollary}}

\newenvironment{proof}[1][Proof]{\begin{trivlist}
\item[\hskip \labelsep {\bfseries #1}]}{\end{trivlist}}

\newcommand{\qed}{\nobreak \ifvmode \relax \else
      \ifdim\lastskip<1.5em \hskip-\lastskip
      \hskip1.5em plus0em minus0.5em \fi \nobreak
      \vrule height0.55em width0.5em depth0.25em\fi}

\usepackage[labelformat=simple]{subcaption}

\long\def\symbolfootnote[#1]#2{\begingroup%
\def\thefootnote{\fnsymbol{footnote}}\footnote[#1]{#2}\endgroup}

\newcount\colveccount
\newcommand*\colvec[1]{
        \global\colveccount#1
        \begin{pmatrix}
        \colvecnext
}
\def\colvecnext#1{
        #1
        \global\advance\colveccount-1
        \ifnum\colveccount>0
                \\
                \expandafter\colvecnext
        \else
                \end{pmatrix}
        \fi
}

\renewcommand{\thepage}{\arabic{page}}
\begin{document}

\title{Velocity-Aware Handover Management in Two-Tier Cellular Networks}

\author{Rabe~Arshad,
        Hesham~ElSawy,
        Sameh~Sorour,
        Tareq~Y.~Al-Naffouri,
        and~Mohamed-Slim~Alouini
\thanks{Rabe Arshad and Sameh Sorour are with the Department
of Electrical Engineering, King Fahd University of Petroleum and Minerals (KFUPM), Saudi Arabia. E-mail: \{g201408420, samehsorour\}@kfupm.edu.sa}
\thanks{Hesham Elsawy, Tareq Y. Al-Naffouri, and Mohamed-Slim Alouini are with King Abdullah University of Science and Technology (KAUST), Computer, Electrical and Mathematical Science and Engineering Division (CEMSE), Thuwal 23955-6900, Saudi Arabia. Email: \{hesham.elsawy, tareq.alnaffouri, slim.alouini\}@kaust.edu.sa}}

\maketitle


\begin{abstract}

While network densification is considered an important solution to cater the ever-increasing capacity demand, its effect on the handover (HO) rate is overlooked. In dense 5G networks, HO delays may neutralize or even negate the gains offered by network densification. Hence, user mobility imposes a nontrivial challenge to harvest capacity gains via network densification. In this paper, we propose a velocity-aware HO management scheme for two-tier downlink cellular network to mitigate the HO effect on the foreseen densification throughput gains. The proposed HO scheme sacrifices the best BS connectivity, by skipping HO to some BSs along the user's trajectory, to maintain longer connection durations and reduce HO rates. Furthermore, the proposed scheme enables cooperative BS service and strongest interference cancellation to compensate for skipping the best connectivity. To this end, we consider different HO skipping scenarios and develop a velocity-aware mathematical model, via stochastic geometry, to quantify the performance of the proposed HO scheme in terms of the coverage probability and user throughput. The results highlight the HO rate problem in dense cellular environments and show the importance of the proposed HO schemes. Finally, the value of BS cooperation along with handover skipping is quantified for different user mobility profiles.
\end{abstract}

\begin{IEEEkeywords}
Multi-tier Dense Cellular Networks; Handover Management; Stochastic Geometry; CoMP; Throughput.\\
\end{IEEEkeywords}
\IEEEpeerreviewmaketitle
\section{Introduction}

\IEEEPARstart{N}{etwork} densification is a potential solution to cater the increasing traffic demands and is expected to have a major contribution in fulfilling the ambitious 1000-fold capacity improvements required for next generation 5G cellular networks~\cite{1}. Network densification improves the spatial frequency reuse by shrinking the BSs' footprints to increase the delivered spatial spectral efficiency. Densifying the network decreases the load served by each BS, and hence, increases the per user throughput. However, such improvement comes at the expense of increased handover (HO) rates for mobile users. Mobile users change their BS associations more frequently in a denser network environment, due to the reduced BSs' footprints, to maintain the best connectivity. The HO procedure involves signaling between the mobile user, serving BS, target BS, and the core network, which consumes physical resources and incurs delay. Therefore, the per user HO rate is always a performance limiting parameter for cellular operators. In extreme cases, where high mobility exists in urban regions, such as users riding monorails in downtowns, ultra-dense cellular networks may fail to support users due to small dwell times within each BS footprint.

 Motivated by the importance of network densification and the significance of the HO problem, several researchers started to exploit stochastic geometry to characterize, understand, and solve the HO problem in dense cellular networks. Stochastic geometry is a powerful mathematical tool that has shown success to characterize the performance of cellular networks with stationary users \cite{7a, 5a,  6a}. Using stochastic geometry, the handover rate in cellular networks is characterized in \cite{lin} for a single tier cellular network with the random waypoint mobility model and in \cite{10a} for a multi-tier cellular network with an arbitrary mobility model. However,~\cite{lin} and \cite{10a} focus only on the HO rate and do not investigate the effect of HO on the throughput. Stochastic geometry models that incorporate handover effect into throughput analysis can be found in \cite{sadr2015handoff, zhangdelay, ge2015user}. However, none of \cite{sadr2015handoff, zhangdelay, ge2015user} propose a solution for the HO problem. The authors in~\cite{cu-split}, propose control plane and user plane split architecture with macro BS anchoring to mitigate the handover effect in dense cellular environment and quantified the performance gain via stochastic geometry. However, the solution proposed in~\cite{cu-split} is not compatible with the current cellular networks and requires massive architectural upgrade to the network.

 In this paper, we propose a simple yet effective velocity aware handover management scheme in a two-tier cellular network that is compatible with the current cellular architecture\footnote{This work has been presented in parts for single tier cellular network in \cite{icc, globe}.}. The proposed scheme, denoted as HO skipping, bypasses association with some BSs along the user's trajectory to maintain a longer service duration with serving BSs and reduce the HO rate and its associated signaling. In other words, the proposed HO skipping scheme sacrifices the best signal-to-interference-plus-noise-ratio (SINR) association to alleviate excessive HO rate and mitigate the handover effect. The proposed scheme also employs interference cancellation (IC) and cooperative BS service, via coordinated multipoint (CoMP) transmission \cite{3gpp, crancomp1, crancomp2}, when the user is not associated to the BS offering the best SINR. When the user decides to skip the best SINR association, denoted as blackout phase, the user is simultaneously served by the BSs that offer the second and third best SINR associations via non-coherent transmission. It is worth mentioning that the non-coherent transmission is considered as it may be hard to estimate the channel state information (CSI) in the considered high mobility scenarios.

 The performance gain of the proposed HO skipping scheme is quantified using stochastic geometry, in which the cellular network is assumed to be spatially deployed according to a Poisson point process (PPP). The PPP assumption is widely accepted for modeling cellular networks and has been verified in \cite{7a,8a,9a} by several empirical studies. To this end, we derive mathematical expressions for the coverage probabilities and the average throughput for the proposed HO skipping schemes. The results manifest the HO problem in dense cellular environments when employing the conventional HO scheme (i.e., best SINR association). Compared to the always best SINR connectivity, the proposed HO schemes show some degradation in the overall coverage probability, but tangible gains are achieved in the terms of average throughput.

\section{System Model}

We consider a two-tier downlink cellular network with CoMP transmission between the BSs belonging to the same or different tiers. It is assumed that the BSs belonging to the $k^{th}$ tier have same transmit power $P_{k}$ and are spatially distributed via a two-dimensional homogenous PPP $\Phi_{k}$ with intensity $\lambda_{k}$, $k\in\{1,2\}$. The macro and femto cell tiers are denoted by $k=1$ and $k=2$, respectively. A power-law path-loss model with path loss exponent $\eta>2$ is considered. For simplicity, we consider the same path loss exponents for the two tiers (i.e. $\eta_{1}=\eta_{2}=\eta$). Extensions to different path-loss exponents is straightforward, however, on the expense of more involved expressions. In addition to path loss, the channel introduces multi-path fading in the transmitted signal. Channel gains are assumed to have Rayleigh distribution with unit power i.e. $h\sim \exp(1)$. Without loss of generality, we conduct our analysis on a test user and assume that all BSs in $\Phi_{1}$ and $\Phi_{2}$ are ascendingly ordered according to their distances from that user. Let $R_i $ and $r_{i}$ be the distances from the test user to the $i^{th}$ BS in $\Phi_{1}$ and $\Phi_{2}$, respectively, then the inequalities $(R_1 < R_2 < ....)$ and $(r_1 < r_2 < ....)$ always hold. We consider a universal frequency reuse scheme and study the performance of one frequency channel. Hence, the best received signal strength (RSS) association implies the best SINR association. A list of key mathematical notations used in this paper is given in table~\ref{not}.
\begin{figure}[!t]
\centering
\hspace{-0.5cm}
\includegraphics[width=0.55 \linewidth]{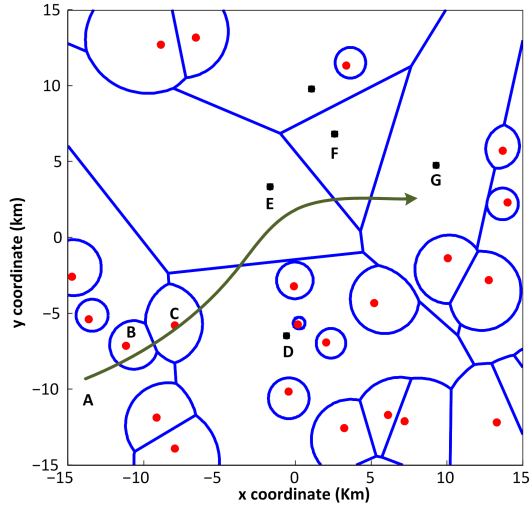}
\small \caption{Voronoi tessellation of a two tier cellular network. Green solid line represents user's trajectory while black squares and red circles represent macro and femto BSs, respectively.}
\label{model}
\end{figure}
\begin{table}[ht]
\centering
\footnotesize
\renewcommand{\arraystretch}{1.3}
\caption{Mathematical Notations}
\label{not}
 \begin{tabular}{||p{1.1cm}|p{5.9cm}||p{1.1cm}|p{5.9cm}||}

\hline  \textbf{Notation} & \textbf{Description} & \textbf{Notation} & \textbf{Description}\\
\hline  $\Phi_{k}$ & PPP of BSs of $k^{th}$ tier & $\eta$ & Path loss exponent\\
\hline  $\lambda_{k}$ & BS intensity of $k^{th}$ tier & $P_{k}$ & Transmit power of BSs of $k^{th}$ tier\\
 \hline $R_{i}$ & Distance between the user and $i^{th}$ macro BS & $r_{i}$ & Distance between the user and $i^{th}$ femto BS\\
 \hline $H_{ij}$ & Handover rate from tier $i$ to $j$ & $\mathcal{R}$ & Achievable rate per unit bandwidth\\
 \hline $d_{m}$ & Macro to macro HO delay  & $d_{f}$ & Femto related HO delay\\
 \hline $\mathcal{C}$ & Coverage probability & $AT$ & Average Throughput\\
 \hline $A_{m}$ & The probability that the macro BS provides the best SINR& $A_{f}$& The probability that the femto BS provides the best SINR\\
 \hline

          \end{tabular}
        \end{table}
        \normalsize
\subsection{User Mobility and Handover Strategies}
In the depicted system model, the best SINR association regions for the BSs can be visualized via a weighted voronoi tessellation as show in Fig.~\ref{model}. Therefore, the conventional scheme executes a HO every time the user crosses a voronoi cell boundary to ensure that the best SINR association is always satisfied. We assume that the test user moves with a constant velocity $v$ on an arbitrary long trajectory that passes through all association and SINR states. The average SINR through a randomly selected user's trajectory is inferred from the stationary PPP analysis. It is worth noting that similar assumption was used in~\cite{sadr2015handoff, zhangdelay, ge2015user, cu-split} for tractability. However, we incorporate user mobility in the simulations and verify the accuracy of the stationary SINR analysis for mobile users. This implies that averaging over all users' trajectories in all network realizations is equivalent to averaging over all users' locations in all network realizations.

 We propose multiple HO skipping strategies that show throughput gains over different user mobility profiles. Particularly, we consider four HO strategies, which represent user mobility profiles ranging from nomadic to high velocities. Mobile users maintain a list of nearby BSs based on the RSS levels and report to the core network through the serving BS. In some cases, HO decisions are made on the radio network level based on the HO type. However, in all cases, the HOs are directed by the network entities, which have the capabilities to trace the user location and velocity using timing advance \cite{TA}. According to the employed HO strategy, the admission controller can help the users to skip the recommended HOs based on their velocities. The BS skipping sacrifices the best SINR connectivity to reduce the handover rate and delay. In order to compensate for the degraded SINR during blackout phases, we enable BS cooperation and IC. For the IC, the interfering signal from the skipped BS is detected, demodulated, decoded and then subtracted from the received signal \cite{16a}. We propose the following HO strategies for the mobile users.\\
\subsubsection{\textbf{Best Connected Strategy {\rm (BC)}}}
In the best connected strategy, the admission controller ensures that the RSS based association is always satisfied for each HO request received from the mobile station. That is, the user is connected the nearest macro BS if $P_{1}R_{1}^{-\eta}>P_{2}r_{1}^{-\eta}$ is satisfied and to the nearest femto BS if $P_{1}R_{1}^{-\eta}<P_{2}r_{1}^{-\eta}$ is satisfied. For the user's trajectory shown in Fig.~\ref{model}, the best connected strategy enforces 6 HOs when the user moves from BS A to G through the BSs \{B, C, D, E, F\}. \\
\subsubsection{ \textbf{Femto Skipping Strategy {\rm (FS)}}}
In the femto skipping strategy, we propose that the user skips some of the femto BSs along its trajectory, when $P_{1}R_{1}^{-\eta}<P_{2}r_{1}^{-\eta}$ is satisfied, to reduce the handover rate. In particular, the user can alternate between the best connectivity and skipping of the femto BSs along its trajectory. During the femto blackout phase, BS cooperation is enabled, which can be intra or inter-tier cooperative BS transmission depending on the relative positions of the BSs along the user's trajectory. For the user's trajectory shown in Fig.~\ref{model}, the \underline{FS} strategy offers 5 HOs (i.e., \{C, D, E, F, G\}) while going from the BS A to G. Also, the user is jointly served by the BSs \{A, C\} while skipping of the BS \{B\}.\\
\subsubsection{\textbf{Femto Disregard Strategy {\rm (FD)}}}
At high mobility profiles, the cell dwell time within the femto BS coverage area may be too small. Hence, we propose the femto disregard strategy where the user skips HOs to the entire femto tier while enabling the cooperative service between the two strongest macros in blackout. This states that the user connects to the nearest macro BS if $P_{1}R_{1}^{\eta}>P_{2}r_{1}^{-\eta}$ and to the first and the second strongest macros, otherwise. For the user's trajectory shown in Fig.~\ref{model}, the \underline{FD} strategy offers 4 HOs (i.e., \{D, E, F, G\}) while going from BS A to G and the joint transmission between the BSs \{A, D\} is enabled while skipping of the BSs \{B, C\}.\\
\subsubsection{\textbf{Macro Skipping Strategy {\rm (MS)}} }
At extremely high velocities, the cell dwell time within the macro BS area may become too small. In this case, in addition to the femto disregard, the user may skip some macro BSs along its trajectory. Particularly, the user alternates between the macro best connectivity and macro blackout phases, where macro BS cooperation in enabled in the macro blackout phase. That is, the user spends $50\%$ of the time in macro best connected mode and rest of $50\%$ in the macro blackout mode. For the user's trajectory shown in Fig.~\ref{model}, the \underline{MS} strategy enforces only 2 HOs (i.e., \{E, G\}), when the user moves from BS A to G and the cooperation is enabled between the BSs \{A, E\} and \{E, G\} while skipping of the macro BSs D and F, respectively.

\subsection{Methodology of Analysis}
We assume that no data is transmitted during HO execution and that the HO duration is dedicated for exchanging control signaling between the serving BS, target BS, and the core network. We consider different backhauling schemes that impose different HO and signaling delays~\cite{14a}. In all cases, the achievable rate is calculated over the time interval where data can be transmitted only. For each of the aforementioned HO skipping strategies, we show the imposed tradeoff between coverage probability and throughput. For the sake of an organized presentation, we show the analysis for each HO strategy in a separate section. In the analysis of each strategy, we first derive the distance distribution between the user and its serving BS as well as the Laplace transform (LT) of the aggregate interference PDF, which are then used to obtain the coverage probability and achievable rate. As discussed earlier, the coverage probabilities and achievable rates are obtained based on the stationary analysis and are verified via simulations in Section III. The handover cost is incorporated to the analysis in Section IV, where the handover rate is calculated and is used to determine the handover delay and average throughput.

\section{Distance Analysis and Coverage Probability}
In this section, we first calculate the service distance distributions for the aforementioned HO skipping cases, which are subsequently used to obtain the coverage probabilities in each case. It is worth noting that the service distance distribution for each HO skipping case is different due to different serving BS(s) in each case. For the sake of an organized presentation, we perform case by case analysis. In the end of this section, we validate the stationary analysis via simulations that account for user mobility for all HO skipping scenarios.

\subsection{\textbf{Best Connected Strategy {\rm (BC)}}} \label{sec3A}
In the best connectivity case, the user associates with the BS that provides the highest power. Thus, the user changes its association when it crosses the boundary of the neighboring cell. The always best connected case has been considerably analyzed in the literature. Here, we follow \cite{5a} and write the distribution of the distances between the user and its serving macro and femto BSs in a two tier network, which is given by the following lemma.
\begin{lem}
In a two tier cellular network, the distance distribution between the user and its serving macro BS is given by
\small
\begin{eqnarray}
f_{R_1}^{(BC)}(R)=\frac{2\pi \lambda_{1} R}{A_{m}^{(BC)}}\exp\left(-\pi R^{2}\left(\lambda_{1}+\lambda_{2}\left(\frac{P_{2}}{P_{1}}\right)^{2/\eta}\right)\right);\quad 0\leq R \leq \infty
\label{f1}
\end{eqnarray}
\normalsize
The distance distribution between the user and its serving femto BS can be expressed as
\small
\begin{eqnarray}
f_{r_1}^{(BC)}(r)=\frac{2\pi \lambda_{2} r}{A_{f}^{(BC)}}\exp\left(-\pi r^{2}\left(\lambda_{2}+\lambda_{1}\left(\frac{P_{1}}{P_{2}}\right)^{2/\eta}\right)\right);\quad 0\leq r \leq \infty
\label{f2}
\end{eqnarray}
\normalsize
where $A_{m}^{(BC)}$ and $A_{f}^{(BC)}$ are the association probabilities for macro and femto BSs, respectively.
\small
\begin{align}
A_{m}^{(BC)}= \frac{\lambda_{1}}{\lambda_{1}+ \lambda_{2}\big(P_{2}/P_{1}\big)^{2/\eta}}, \quad A_{f}^{(BC)}= \frac{\lambda_{2}}{\lambda_{2}+ \lambda_{1}\big(P_{1}/P_{2}\big)^{2/\eta}}.
\label{Af}
\end{align}
\normalsize
\label{case1dist}
\end{lem}
\begin{proof}
The lemma is obtained by using the same methodology as shown in~\cite[Lemma 3]{5a} but considering same path loss exponent and unity bias factor for both tiers.
\end{proof}
\subsubsection{Coverage Probability}
The coverage probability is defined as the probability that the received SINR exceeds some threshold $T$. In case \underline{BC}, the user associates with the macro BS with probability $A_{m}^{(BC)}$ and with the femto BS with probability $A_{f}^{(BC)}$, where the association is based on the highest RSS. By the law of total probability, we can write the overall coverage probability as
\small
\begin{eqnarray}
\mathcal{C}^{(BC)}=A_{m}^{(BC)}\mathcal{C}_{m}^{(BC)}+A_{f}^{(BC)}\mathcal{C}_{f}^{(BC)},
\end{eqnarray}
\normalsize
where $\mathcal{C}_{m}^{(BC)}$ and $\mathcal{C}_{f}^{(BC)}$ are the coverage probabilities for the serving macro and femto BSs, respectively. The coverage probabilities $\mathcal{C}_{m}^{(BC)}$ and $\mathcal{C}_{f}^{(BC)}$ are given by:
\small
\begin{eqnarray}
\mathcal{C}_{m}^{(BC)}= \mathbb{P}\bigg[\frac{P_{1} h R_{1}^{-\eta}}{I_{R(m)}+I_{r(m)}+\sigma^{2}}>T\bigg], \quad \mathcal{C}_{f}^{(BC)}=\mathbb{P}\bigg[\frac{P_{2} h r_{1}^{-\eta}}{I_{R(f)}+I_{r(f)}+\sigma^{2}}>T\bigg],
\end{eqnarray}
\normalsize
where $I_{R(\cdot)}$ and $I_{r(\cdot)}$ are the aggregate interference powers received from the macro and femto tiers, respectively, which are defined as
\small
\begin{align*}
I_{R(m)}= \sum_{i\epsilon\phi_1\backslash b_{1}}^{} P_{1}h_{i}R_{i}^{-\eta}, \quad I_{r(m)}= \sum_{i\epsilon\phi_2}^{} P_{2}h_{i}r_{i}^{-\eta},\quad I_{R(f)}= \sum_{i\epsilon\phi_1}^{} P_{1}h_{i}R_{i}^{-\eta},\quad I_{r(f)}= \sum_{i\epsilon\phi_2\backslash b_{1}}^{} P_{2}h_{i}r_{i}^{-\eta}.
\end{align*}
\normalsize
Following \cite{6a}, conditioning on the distance between the user and the serving BS and exploiting the exponential distribution of $h_{i}$, the conditional coverage probabilities are given by
\small
\begin{eqnarray}
\mathcal{C}_{m}^{(BC)}(R_{1})=\exp\bigg(\frac{-TR_{1}^{\eta}\sigma^{2}}{P_{1}}\bigg)\mathscr{L}_{I_{R}(m)}\bigg(\frac{TR_{1}^{\eta}}{P_{1}}\bigg)\mathscr{L}_{I_{r}(m)}\bigg(\frac{TR_{1}^{\eta}}{P_{1}}\bigg),
\label{cond1}
\end{eqnarray}
\begin{eqnarray}
\mathcal{C}_{f}^{(BC)}(r_{1})=\exp\bigg(\frac{-Tr_{1}^{\eta}\sigma^{2}}{P_{2}}\bigg)\mathscr{L}_{I_{R}(f)}\bigg(\frac{Tr_{1}^{\eta}}{P_{2}}\bigg)\mathscr{L}_{I_{r}(f)}\bigg(\frac{Tr_{1}^{\eta}}{P_{2}}\bigg),
\label{cond12}
\end{eqnarray}
\normalsize
where $\mathcal{C}_{m}^{(BC)}(R_{1})$ and $\mathcal{C}_{f}^{(BC)}(r_{1})$ are the conditional coverage probabilities for macro and femto associations, respectively. The LTs of $I_{R}$ and $I_{r}$ for the macro and femto association cases are evaluated in the following lemma.
\begin{lem}
The Laplace transforms of $I_{R}$ and $I_{r}$ in the macro association case are given by
\small
\begin{eqnarray}
\mathscr{L}_{I_{R}(m)}(s)=\exp\bigg(-\frac{2\pi\lambda_{1}TR_{1}^{2}}{\eta-2}\mathstrut_2 F_1\left(1,1-\frac{2}{\eta},2-\frac{2}{\eta},-T\right)\bigg),
\label{LT1}
\end{eqnarray}
\begin{align}
\hspace{-0.1cm}\mathscr{L}_{I_{r}(m)}(s)=\exp\bigg(\hspace{-0.1cm}-\frac{2\pi\lambda_{2}TR_{1}^{2}}{\eta-2}\left(\frac{P_{2}}{P_{1}}\right)^{2/\eta}\mathstrut_2 F_1\left(1,1-\frac{2}{\eta},2-\frac{2}{\eta},-T\right)\bigg),
\label{l2}
\end{align}
\normalsize
The Laplace transforms of $I_{R}$ and $I_{r}$ in the femto association case can be expressed as
\small
\begin{align}
\hspace{-0.15cm}\mathscr{L}_{I_{R}(f)}(s)=\exp\bigg(\hspace{-0.1cm}-\frac{2\pi\lambda_{1}Tr_{1}^{2}}{\eta-2}\left(\frac{P_{1}}{P_{2}}\right)^{2/\eta}\mathstrut_2 F_1\left(1,1-\frac{2}{\eta},2-\frac{2}{\eta},-T\right)\bigg),
\label{l3}
\end{align}
\begin{align}
\mathscr{L}_{I_{r}(f)}(s)=\exp\bigg(-\frac{2\pi\lambda_{2}Tr_{1}^{2}}{\eta-2}\mathstrut_2 F_1\left(1,1-\frac{2}{\eta},2-\frac{2}{\eta},-T\right)\bigg),
\label{l4}
\end{align}
\normalsize
where $\mathstrut_2F_1\big(.,.,.,.\big)$ is a hypergeometric function.
\label{case1}
\end{lem}
\begin{proof}
See Appendix A.
\end{proof}
In the special case when $\eta=4$, which is a common path loss exponent for outdoor environments, the LTs in~\eqref{LT1}-\eqref{l4} boil down to much simpler expressions as shown in the following corollary.
\begin{col}
For the special case of $\eta=4$ the LTs for the macro association case given in Lemma~\ref{case1} reduce to
\small
\begin{eqnarray}
\mathscr{L}_{I_{R}(m)}(s)|_{\eta=4}= \exp\left(- \pi\lambda_{1} R_{1}^{2} \sqrt{T}\arctan\left({\sqrt{T}}\right) \right),
\end{eqnarray}
\begin{eqnarray}
\mathscr{L}_{I_{r}(m)}(s)|_{\eta=4}= \exp\left(- \pi\lambda_{2} R_{1}^{2} \sqrt{\frac{TP_{2}}{P_{1}}}\arctan\left({\sqrt{T}}\right) \right).
\end{eqnarray}
\normalsize
The LTs for the femto association case evaluated at $\eta=4$ are given by
\small
\begin{eqnarray}
\mathscr{L}_{I_{R}(f)}(s)|_{\eta=4}= \exp\left(- \pi\lambda_{1} r_{1}^{2} \sqrt{\frac{TP_{1}}{P_{2}}}\arctan\left({\sqrt{T}}\right) \right),
\end{eqnarray}
\begin{eqnarray}
\mathscr{L}_{I_{r}(f)}(s)_{\eta=4}= \exp\left(- \pi\lambda_{2} r_{1}^{2} \sqrt{T}\arctan\left({\sqrt{T}}\right) \right).
\end{eqnarray}
\normalsize
\end{col}
Combining Lemmas~\ref{case1dist} and~\ref{case1}, the following theorem is obtained for the coverage probability.
\begin{theorem}
Considering two independent PPPs based two tier cellular network with BS intensity $\lambda_{i}$ in a Rayleigh fading environment, the coverage probabilities for the macro and femto users can be written as\\
\scriptsize
\begin{align}
\mathcal{C}_{m}^{(BC)}= \frac{2\pi\lambda_{1}}{A_{m}^{(BC)}}\int_{0}^{\infty}R_{1}\exp\left(-\frac{T R_{1}^{\eta}\sigma^{2}}{P_{1}}-\pi R_{1}^{2}\left(\lambda_{1}+\lambda_{2}\left(\frac{P_{2}}{P_{1}}\right)^{\frac{2}{\eta}}\right) \bigg(1+\frac{2 T }{\eta-2}\mathstrut_2 F_1\left(1,1-\frac{2}{\eta},2-\frac{2}{\eta},-T\right)\bigg)\right)dR_{1}.
\label{c1}
\end{align}
\begin{align}
\mathcal{C}_{f}^{(BC)}= \frac{2\pi\lambda_{2}}{A_{f}^{(BC)}}\int_{0}^{\infty}r_{1}\exp\left(-\frac{T r_{1}^{\eta}\sigma^{2}}{P_{2}}-\pi r_{1}^{2}\left(\lambda_{2}+\lambda_{1}\left(\frac{P_{1}}{P_{2}}\right)^{\frac{2}{\eta}}\right)\bigg(1+\frac{2 T }{\eta-2}\mathstrut_2 F_1\left(1,1-\frac{2}{\eta},2-\frac{2}{\eta},-T\right)\bigg)\right)dr_{1}.
\label{c12}
\end{align}
\normalsize
\label{theorem1}
\end{theorem}
\begin{proof}
The theorem is proved by substituting the LTs obtained in Lemma~\ref{case1} in the conditional coverage probability expressions given in~\eqref{cond1} and~\eqref{cond12} and then integrating over the service distance distributions provided by Lemma~\ref{case1dist}.
\end{proof}
In an interference limited environment with path loss exponent $\eta=4$, the coverage probabilities in Theorem~\ref{theorem1} simplify to the following closed form expressions.
\small
\begin{eqnarray}
\mathcal{C}_{m}^{(FS)}=\mathcal{C}_{f}^{(BC)}=\frac{1}{1+\sqrt{T}\arctan{(\sqrt{T})}}.
\end{eqnarray}
\normalsize
\subsection{\textbf{Femto Skipping Strategy {\rm (FS)}}}
In the femto skipping case, the test user associates with the macro BS based on the highest RSS. However, the user skips some femto BS associations to reduce excessive HO rate, where the user experience blackout during femto skipping. In the blackout phase, the user is simultaneously served by the second and the third strongest BSs via non-coherent CoMP transmission. The cooperating BSs can be both macros, both femtos, or one macro and one femto. We assume that the user alternates between the femto best connected and femto blackout phases. The service distance distributions for the best connectivity associations (i.e., non-blackout) in \underline{FS} scheme are similar to that of \underline{BC} scheme given in~\eqref{f1} and~\eqref{f2} i.e. $f_{R_{1}}^{(FS)}=f_{R_{1}}^{(BC)}$ and $f_{r_{1}}^{(FS)}=f_{r_{1}}^{(BC)}$. However, in the blackout case, the distance distributions are different and have to be derived for each pair of cooperating base stations (i.e. macro and macro, femto and femto, macro and femto). Furthermore, the coverage probability in the blackout case is different for each of the cooperating BSs case and the probability of each cooperation event should be calculated to obtain the total coverage probability. Therefore, it is cumbersome to derive the distance distributions and coverage probabilities while accounting for the cooperative BSs types via the conventional procedure used in the literature and shown in Section~\ref{sec3A}. Instead, we follow~\cite{shotgun} and exploit the mapping theorem to develop a unified analysis for all cooperation instances by mapping the two dimensional PPPs into an equivalent one dimensional non-homogenous PPP.
\begin{lem} \label{mapping}
The two point processes $\Phi_1$ and $\Phi_2$ seen from the test receiver perspective is statistically equivalent to a one dimensional non-homogeneous PPP with intensity
\small
\begin{equation}
\lambda(y)= \frac{2 \pi}{\eta} \left(\lambda_1 P_1^{2/\eta} + \lambda_2 P_2^{2/\eta}\right) y^{2/\eta-1}.
\end{equation}
\normalsize
\end{lem}
\begin{proof}
See Appendix B.
\end{proof}
Using Lemma~\ref{mapping}, we do not need to account for the cooperating BSs types and are able to derive a unified distance distribution and coverage probability expression that accounts for all cooperation instances. This is demonstrated in the following lemma.
\begin{lem}
Let $x$ and $y$ be the distances between the user and the cooperating BSs. Conditioning on $x$, the conditional distance distribution of the skipped BS with distance $r_1$ from the user conditioned on the second nearest BS in the blackout case is given by
\small
\begin{eqnarray}
f_{r(bk)}^{(FS)}(r_1|x)=\frac{2r_{1}^{2/\eta-1}}{x^{2/\eta}};\quad 0\leq r_{1}\leq x \leq \infty
\label{20}
\end{eqnarray}
\normalsize
and the joint distance distribution of $x$ and $y$ is given by

\small
\begin{eqnarray}
\hspace{-0.2cm}f_{X,Y(bk)}^{(FS)}(x,y)=\frac{4}{\eta^{2}}\big(\pi\lambda_{t})^{3}x^{4/\eta-1}y^{2/\eta-1}\exp(-\pi \lambda_{t}y^{2/\eta});\quad 0\leq x\leq y \leq \infty
\end{eqnarray}
\normalsize
where
\small
\begin{eqnarray}
\lambda_{t}=\lambda_{1}P_{1}^{2/\eta}+\lambda_{2}P_{2}^{2/\eta}.
\end{eqnarray}
\normalsize
\label{dist2}
\end{lem}
\begin{proof}
See Appendix C.
\end{proof}
\subsubsection{Coverage Probability}
By the law of total probability, the overall coverage probability for the case \underline{FS} is given by
\small
\begin{align}
\mathcal{C}^{(FS)}=A_{m(\bar{bk})}^{(FS)}\mathcal{C}_{m(\bar{bk})}^{(FS)}+A_{f(\bar{bk})}^{(FS)}\mathcal{C}_{f(\bar{bk})}^{(FS)}+A_{bk}^{(FS)}\mathcal{C}_{bk}^{(FS)}.
\end{align}
\normalsize
where $\bar{bk}$ and $bk$ represent the non-blackout and blackout phases, respectively. The coverage probabilities for the macro and femto associations in the non-blackout case are the same as the probabilities derived in~\eqref{c1} and~\eqref{c12} i.e., $\mathcal{C}_{m(\bar{bk})}^{(FS)}=\mathcal{C}_{m}^{(BC)}$ and $\mathcal{C}_{f(\bar{bk})}^{(FS)}=\mathcal{C}_{f}^{(BC)}$. Also, the macro association probability $A_{m(\bar{bk})}^{(FS)}$ is the same as derived in Lemma~\ref{case1dist}. In the blackout phase, the user skips the strongest femto BS candidate and is served by the second and the third strongest BSs via non-coherent CoMP, which changes the blackout coverage probability and results in two femto association probabilities (i.e., blackout and non-blackout associations). Since, the user alternates between the femto best connected and femto blackout phases, the probabilities that the user is in femto best connected (non-blackout) and blackout phases can be expressed as $A_{f(\bar{bk})}^{(FS)}=A_{bk}^{(FS)}= 0.5 {A_{f}^{(BC)}}$. By employing the mapping theorem given in Lemma 3, we can lump the aggregate interference from both tiers and express the coverage probability in the blackout phase as
\small
\begin{eqnarray}
\mathcal{C}_{X,Y(bk)}^{(BC)}=\mathbb{P}\bigg[\frac{\left|h_{1}x^{-\frac{1}{2}}\hspace{-0.1cm}+h_{2} y^{-\frac{1}{2}}\right|^2}{I_{agg}+I_{r_{1}}+\sigma^2}>T\bigg],
\end{eqnarray}
\normalsize
where
\small
\begin{eqnarray}
I_{r_{1}}=h_{1}\sqrt{r_{1}},\quad I_{agg}=\sum_{i\epsilon\phi\backslash b_{1}\text{, }b_{2}\text{, }b_{3}}^{}h_{i}\sqrt{ z_{i}},
\end{eqnarray}
\normalsize
where $r_{1}$ represents the distance between the user and the skipped femto BS while $z_{i}$ represents the distance between the user and the interfering BSs belonging to both tiers.
Since, $h_{i}'s$ are $i.i.d.$ $\mathcal{CN}(0,1)$, such that $|x_{1}h_1+x_{1}h_2|^2\sim \exp(\frac{1}{x_{1}^{2}+x_{2}^{2}})$, we can write the conditional coverage probability (conditioned on the serving BSs) as
\small
\begin{eqnarray}
\mathcal{C}_{X,Y(bk)}^{(FS)}(x,y)=\exp\bigg(\hspace{-0.05cm}\frac{-T\sigma^2}{x^{-1}\hspace{-0.1cm}+\hspace{-0.1cm}y^{-1}}\hspace{-0.05cm}\bigg)\mathscr{L}_{I_{r_{1}}}\hspace{-0.05cm}\bigg(\hspace{-0.05cm}\frac{T}{x^{-1}\hspace{-0.1cm}+\hspace{-0.1cm}y^{-1}}\bigg)\mathscr{L}_{I_{agg}}\bigg(\hspace{-0.05cm}\frac{T}{x^{-1}\hspace{-0.1cm}+\hspace{-0.1cm}y^{-1}}\hspace{-0.05cm}\bigg)
\label{cond2}
\end{eqnarray}
\normalsize
The Laplace transforms of $I_{r_{1}}$ and $I_{agg}$ are given by the following lemma.
\begin{lem}
The Laplace transform of $I_{r_{1}}$ can be expressed as
\small
\begin{eqnarray}
\mathscr{L}_{I_{r_{1}}}(s)=\int_{0}^{x}\frac{2r_{1}^{2/\eta-1}}{\eta x^{2/\eta}\big(1+\frac{T}{r_{1}(x^{-1}+y^{-1})}\big)}dr_{1}.
\label{lt21}
\end{eqnarray}
\normalsize
The Laplace transform of $I_{agg}$ can be expressed as
\small
\begin{eqnarray}
\mathscr{L}_{I_{agg}}(s)=\exp\bigg(\frac{-2\pi\lambda_{t}Ty^{2/n-1}}{(\eta-2)(x^{-1}+y^{-1})}\mathstrut_2 F_1\bigg(1,1-\frac{2}{\eta},2-\frac{2}{\eta},\frac{-T}{x^{-1}y+1}\bigg)\bigg).
\label{lt22}
\end{eqnarray}
\normalsize
\label{lt2}
\end{lem}
\begin{proof}
See Appendix D.
\end{proof}
For the special case ($\eta=4$), the LTs in~\eqref{lt21} and~\eqref{lt22} simplify to the expressions as given by the following corollary.
\begin{col}
The LTs in Lemma~\ref{lt2}, when evaluated at $\eta=4$, reduce to
\small
\begin{eqnarray}
\mathscr{L}_{I_{r_1}}(s)\big|_{\eta=4}=1-\sqrt{\frac{Ty}{x+y}}\arctan{\left(\sqrt{\frac{x+y}{Ty}}\right)},
\end{eqnarray}
\begin{eqnarray}
\mathscr{L}_{I_{agg}}(s)\big|_{\eta=4}=\exp\bigg(-\pi\lambda_{t}\sqrt{\frac{T}{x^{-1}+y^{-1}}}\arctan{\sqrt{\frac{Tx}{x+y}}}\bigg),
\end{eqnarray}
\normalsize
\end{col}
Using the service distance distribution and the LTs derived in Lemma~\ref{dist2} and~\ref{lt2}, the following theorem for the coverage probability is obtained.
\begin{theorem}
\label{case3t}
Considering two independent PPPs based two tier downlink cellular network with BS intensity $\lambda_{i}$ in a Rayleigh fading environment, the coverage probability for the blackout users in case \underline{FS} is obtained and given in~\eqref{cp2}.
\begin{figure*}[h]
\scriptsize
\begin{eqnarray}
\mathcal{C}_{(bk)}^{(FS)}=\int_{0}^{\infty}\int_{x}^{\infty}\frac{8}{\eta^{3}}\pi\lambda_{t}^{3}x^{2/\eta-1}y^{2/\eta-1}e^{-\pi y^{2/\eta}\lambda_{t}-\frac{2\pi\lambda_{t}Ty^{2/n-1}}{(\eta-2)(x^{-1}+y^{-1})}\mathstrut_2 F_1\bigg(1,1-\frac{2}{\eta},2-\frac{2}{\eta},\frac{-T}{x^{-1}y+1}\bigg)}\int_{0}^{x}\frac{r_{1}^{2/\eta-1}}{1+\frac{Tr_{1}^{-1}}{x^{-1}+y^{-1}}}dr_{1} dy dx.
\label{cp2}
\end{eqnarray}
\normalsize
\hrulefill
\end{figure*}
\end{theorem}
\begin{proof}
The theorem is obtained by substituting the LTs shown in Lemma~\ref{lt2} in the conditional coverage probability expression~\eqref{cond2} and integrating over the distance distribution obtained in Lemma~\ref{dist2}.
\end{proof}
\subsubsection{Interference Cancellation}
In the blackout phase, the interference from the skipped BS (i.e., $I_{r_{1}}$) may be overwhelming to the SINR. Hence, interference cancellation techniques could be employed to improve the coverage probability. By considering skipped BS interference cancellation, the coverage probability for the blackout user is given by the following theorem.
\begin{theorem}
\label{theorem:Case2_IC}
Considering an independent PPP based two-tier cellular network with BS intensity $\lambda_{i}$ in a Rayleigh fading environment, the coverage probability for blackout users with interference cancellation capabilities can be expressed as
\small
\begin{align}
\mathcal{C}_{(bk,IC)}^{(FS)}=\int_{0}^{\infty}\int_{x}^{\infty}\frac{4}{\eta^{2}}\pi\lambda_{t}^{3}x^{4/\eta-1}y^{2/\eta-1}e^{-\pi y^{2/\eta}\lambda_{t}-\frac{2\pi\lambda_{t}Ty^{2/n-1}}{(\eta-2)(x^{-1}+y^{-1})}\mathstrut_2 F_1\bigg(1,1-\frac{2}{\eta},2-\frac{2}{\eta},\frac{-T}{x^{-1}y+1}\bigg)} dy dx.
\label{eq:case 2 IC}
\end{align}
\normalsize
\end{theorem}
\begin{proof}
The theorem is obtained using the same methodology for obtaining~\eqref{cp2} but with eliminating $I_{r_{1}}$ from \eqref{cond2}.
\end{proof}
\subsection{\textbf{Femto Disregard Strategy {\rm (FD)}}}
In case \underline{FD}, the mobile user skips all femto BSs associations. Since, the femto BS footprint is quite smaller than the macro BS footprint, we propose that the test user associates with the macro BSs only. Particularly, the user associates with the nearest macro BS if $P_{1}R_{1}^{\eta}>P_{2}r_{1}^{-\eta}$ is satisfied while the cooperative BS service is enabled in blackout (i.e. $P_{1}R_{1}^{\eta}<P_{2}r_{1}^{-\eta}$). Since the femto tier is disregarded, only macro BS cooperation is allowed to compensate for the SINR degradation in blackout.
In non-blackout case, the user associates with the macro BS offering highest RSS, thus, the distance distribution in this case is the same as in case \underline{BC} given in~\eqref{f1} i.e. $f_{R_{1}}^{(FD)}=f_{R_{1}}^{(BC)}$. However, the conditional and the joint PDFs of the distances between the blackout user and its serving macros are given by the following lemma.
\begin{lem}
The conditional distance distribution of the skipped femto BS, conditioned on the serving macro BS, is given by
\small
\begin{eqnarray}
f_{r}^{(FD)}(r_{1}|R_{1})=\frac{2\pi\lambda_{2}r_{1}\exp(-\pi\lambda_{2}r_{1}^{2})}{1-\exp(-\pi\lambda_{2}R_{1}^{2}(\frac{P_{2}}{P_{1}})^{2/\eta})};\quad 0\leq r_{1}\leq \bigg(\frac{P_{2}}{P_{1}}\bigg)^{1/\eta}R_{1} \leq \infty
\label{cond31}
\end{eqnarray}
\normalsize
The joint distance distribution between the test user and its skipped and serving BSs in the cooperative blackout mode can be expressed as
\small
\begin{eqnarray}
f_{R_{1},R_{2},r_{1}(bk)}^{(FD)}(x,y,z)=\frac{(2\pi)^{3}}{A_{bk}^{(FD)}}{\lambda_{1}}^{2}\lambda_{2}x y z \exp\left(-\pi(\lambda_{1}y^{2}+\lambda_{2}z^{2})\right).
\label{jointdist3}
\end{eqnarray}
\normalsize
The marginal distribution of the distance between the user and its first and second strongest macro BSs in the cooperative blackout mode is given by
\small
\begin{align}
f_{R_{1},R_{2}(bk)}^{(FD)}(x,y)=\frac{(2\pi \lambda_{1})^{2}}{A_{bk}^{(FD)}}xy\exp\big(-\lambda_{1}\pi y^{2}\big)
\left(1-\exp\left(-\lambda_{2}\pi x^{2}\left(\frac{P_{2}}{P_{1}}\right)^{2/\eta}\right)\right),
\end{align}
\normalsize
where
\small
\begin{eqnarray}
A_{bk}^{(FD)}=A_{f}^{(BC)}=\frac{\lambda_{2}}{\lambda_{2}+ \lambda_{1}\big(P_{1}/P_{2}\big)^{2/\eta}}.
\end{eqnarray}
\normalsize
\label{dist3}
\end{lem}
\begin{proof}
The conditional distribution $f_{r}^{(FD)}(r_{1}|R_{1})$ is obtained by first writing the joint distribution of $r_{1}$ and $R_{1}$ (i.e., $f_{r,R}(r_{1},R_{1})=(2\pi)^{2}\lambda_{1}\lambda_{2}e^{-\pi\lambda_{1}R_{1}^{2}-\pi\lambda_{2}r_{1}^{2}}$. Then dividing it by the marginal distribution of $R_{1}$, $(r_{1}(\frac{P_{1}}{P_{2}})^{1/\eta}<R_{1}< \infty)$, we obtain $f_{r}(r_{1}|R_{1})$. The joint distribution $f_{R_{1},R_{2},r_{1}(bk)}^{(FD)}(.,.,.)$ is obtained by first writing the conditional PDF of $R_{2}$ conditioning on $R_{1}$ as $f_{R_{2}}(y|R_{1})=2\pi\lambda y e^{-\lambda\pi(y^{2}-R_{1}^{2})}$ and calculating the joint PDF $f_{R_1,R_2}(x,y)=f(y|x)f_{R_1}(x)$. Then, multiplying by the weighted distribution of $r_{1}$ (i.e., using null probability of PPP), we get the joint distribution as given in~\eqref{jointdist3}. The marginal distance distribution between the user and its serving BSs $f_{R_{1},R_{2}(bk)}^{(FD)}(.,.)$ is obtained by integrating~\eqref{jointdist3} with respect to $r_1$, which is bounded from 0 to $R_{1}(\frac{P_{2}}{P_{1}})^{1/\eta}$.
\end{proof}
\subsubsection{Coverage Probability}
By employing the law of total probability, the overall coverage probability in the \underline{FD} scheme can be written as
\small
\begin{eqnarray}
\mathcal{C}^{(FD)}=A_{\bar{bk}}^{(FD)}\mathcal{C}_{m(\bar{bk})}^{(FD)}+A_{bk}^{(FD)}\mathcal{C}_{m,m(bk)}^{(FD)}.
\end{eqnarray}
\normalsize
The event probabilities in the above equation are the same as in case \underline{BC}, given in~\eqref{Af} (i.e., $A_{\bar{bk}}^{(FD)}=A_{m}^{(BC)}$ and $A_{bk}^{(FD)}=A_{f}^{(BC)}$). The coverage probability in the non-blackout case is the same as the macro association probability in case \underline{BC} i.e. $\mathcal{C}_{m(\bar{bk})}^{(FD)}=\mathcal{C}_{m}^{(BC)}$. However, the coverage probability for the blackout case is different from the previous cases and is given by
\small
\begin{eqnarray}
 \mathcal{C}_{m,m(bk)}^{(FD)}= \mathbb{P}\bigg[\frac{|\sqrt{P_{1}}h_{1} R_{1}^{-\frac{\eta}{2}}\hspace{-0.1cm}+\sqrt{P_{1}}h_{2} R_{2}^{-\frac{\eta}{2}}|^2}{I_{R}+I_{r_{1}}+I_{r}+\sigma^2}>T\bigg].
\end{eqnarray}
\normalsize
where, $I_{R}$ and $I_{r}$ are the aggregate interference powers from the macro and femto tiers, respectively. $I_{r_{1}}$ is the interference power from the strongest femto BS. Here, it is worth noting that $I_{R}$ is the received aggregate interference from all macro BSs except $\{b_{1},b_{2}\}$ and $I_{r}$ is the aggregate interference power received from all femtos except $\{b_{1}\}$. The conditional coverage probability (conditioned on the serving BSs) for the blackout case is given by
\small
\begin{align}
\mathcal{C}_{m,m(bk)}^{(FD)}(R_1,R_2)\hspace{-0.1cm}=\exp\bigg(\hspace{-0.05cm}\frac{-T\sigma^2}{x_{1}^{2}\hspace{-0.1cm}+\hspace{-0.1cm}x_{2}^{2}}\hspace{-0.05cm}\bigg)\mathscr{L}_{I_R}\hspace{-0.05cm}\bigg(\hspace{-0.05cm}\frac{T}{x_{1}^{2}\hspace{-0.1cm}+\hspace{-0.1cm}x_{2}^{2}}\bigg)\mathscr{L}_{I_{r_{1}}}\bigg(\hspace{-0.05cm}\frac{T}{x_{1}^{2}\hspace{-0.1cm}+\hspace{-0.1cm}x_{2}^{2}}\hspace{-0.05cm}\bigg)\mathscr{L}_{I_{r}}\bigg(\hspace{-0.05cm}\frac{T}{x_{1}^{2}\hspace{-0.1cm}+\hspace{-0.1cm}x_{2}^{2}}\hspace{-0.05cm}\bigg),
\label{cond3}
\end{align}
\normalsize
where
\vspace{-0.1cm}
\small
\begin{eqnarray*}
x_{i}=\sqrt{P_{i}} R_{i}^{-\eta/2}\hspace{0.2cm}\text{and $h_{i}'s$ are $i.i.d.$ $\mathcal{CN}$(0,1)}.
\end{eqnarray*}
\normalsize
Since, the user is in blackout, the condition $P_{2}r_{1}^{-\eta}>P_{1}R_{1}^{-\eta}$ is satisfied. This implies that the first nearest femto BS must exist between 0 and $R_{1}(\frac{P_{2}}{P_{1}})^{1/\eta}$. Therefore, we derive the LT of $I_{r}$ considering that the interfering femto BSs exist outside an interference exclusion circle with radius $r_{1}$ centered at the test receiver. The LTs of $I_{R}$, $I_{r_{1}}$ and $I_{r}$ in the blackout case are evaluated in the following lemma.
\begin{lem}
The LT of the aggregate interference power received from the macro tier in the blackout phase can be characterized as
\small
\begin{align}
\mathscr{L}_{I_R}(s)=\exp\bigg(\frac{-2\pi T\lambda_{1}R_{2}^{2-\eta}}{(\eta-2)(R_{1}^{-\eta}+R_{2}^{-\eta})}\mathstrut_2 F_1\bigg(1,1-\frac{2}{\eta},2-\frac{2}{\eta},\frac{-TR_{2}^{-\eta}}{R_{1}^{-\eta}+R_{2}^{-\eta}}\bigg)\bigg).
\end{align}
\normalsize
The LT of $I_{r_{1}}$ can be expressed as
\small
\begin{align}
\mathscr{L}_{I_{r_1}}(s)=\int_{0}^{R_{1}(\frac{P_{2}}{P_{1}})^{1/\eta}}\frac{2\pi \lambda_{2}r_{1}e^{-\pi\lambda_{2}r_{1}^{2}}}{\bigg(1+\frac{TP_{1}^{-1}P_{2}r_{1}^{-\eta}}{R_{1}^{-\eta}+R_{2}^{-\eta}}\bigg)\bigg(1-e^{-\lambda_{2}\pi R_{1}^{2}(\frac{P_{2}}{P_{1}})^{2/\eta}}\bigg)}dr_{1}.
\label{Ir1}
\end{align}
\normalsize
The LT of the aggregate interference power received from the entire femto tier except $\{b_{1}\}$ is given by
\small
\begin{align} \mathscr{L}_{I_r}(s)=\exp\bigg(\frac{-2\pi\lambda_{2}P_{2}T}{(\eta-2)P_{1}(R_{1}^{-\eta}+R_{2}^{-\eta})}r_{1}^{2-\eta}\mathstrut_2 F_1\bigg(1,1-\frac{2}{\eta},2-\frac{2}{\eta},\frac{-P_{2}Tr_{1}^{-\eta}}{P_{1}(R_{1}^{-\eta}+R_{2}^{-\eta})}\bigg)\bigg).
\label{Ir}
\end{align}
\normalsize
\label{case3}
\end{lem}
\begin{proof}
The LT of $I_{R}$ is derived using the same procedure as shown for $\mathscr{L}_{I_R(m)}(s)$ in~\eqref{LT1} but considering the macros aggregate interference from $R_{2}$ to $\infty$ and $s=\frac{T}{P_{1}(R_{1}^{-\eta}+R_{2}^{-\eta})}$. Also, the LT of $I_{r_1}$ is derived in a similar way as eq. (8) in~\cite{icc} but considering different $s$ (shown above) and the conditional service distribution shown in~\eqref{cond31}. The LT of $I_{r}$ is derived in a similar way as $\mathscr{L}_{I_R}(s)$ while taking the femto aggregate interference from $r_1$ to $\infty$.
\end{proof}
We evaluate the LTs in Lemma~\ref{case3} for a special case at $\eta=4$, which are given by the following corollary.
\begin{col}
For a special case, when $\eta=4$, the LT of $I_{R}$ boils down to a closed form expression as shown below.
\small
\begin{eqnarray}
\mathscr{L}_{I_R}(s)\big|_{\eta=4}=\exp\left(-\pi\lambda_{1}\sqrt{\frac{T}{R_{1}^{-4}+R_{2}^{-4}}}\arctan{\left(\sqrt{\frac{T R_{1}^{4}}{R_{1}^{4}+R_{2}^{4}}}\right)}\right).
\end{eqnarray}
\normalsize
The LT of $I_{r}$ evaluated at $\eta=4$ is boiled down as follows
\small
\begin{eqnarray}
\mathscr{L}_{I_r}(s)\big|_{\eta=4}=\exp\left(-\pi\lambda_{2}\sqrt\frac{P_{2}P_{1}^{-1}T}{R_{1}^{-4}+R_{2}^{-4}}\arctan\left(\sqrt\frac{P_{2}P_{1}^{-1}Tr_{1}^{-4}}{R_{1}^{-4}+R_{2}^{-4}}\right)\right).
\end{eqnarray}
\normalsize
\end{col}
Using the LTs and distance distributions found in Lemmas~\ref{dist3} and~\ref{case3}, we obtain the final coverage probability for the blackout users in case \underline{FD} as given in Theorem~\ref{case3t}.
\begin{theorem}
\label{case3t}
\begin{figure*}
\scriptsize
\begin{multline}
\mathcal{C}_{m,m(bk)}^{(FD)}=\int_{0}^{\infty}\int_{R_{1}}^{\infty}\int_{0}^{R_{1}(\frac{P_{2}}{P_{1}})^{1/\eta}}\frac{(2\pi)^{3}\lambda_{1}^{2} \lambda_{2}r_{1}R_{1}R_{2}}{\bigg(1+\frac{TP_{2}r_{1}^{-\eta}}{P_{1}(R_{1}^{-\eta}+R_{2}^{-\eta})}\bigg)A_{bk}^{(FD)}\bigg(1-e^{-\lambda_{2}\pi R_{1}^{2}(\frac{P_{2}}{P_{1}})^{2/\eta}}\bigg)}\exp\Bigg(-\pi \lambda_{2}r_{1}^{2}-\pi\lambda_{1}R_{2}^{2}-\\ \frac{T\sigma^{2}}{P_{1}(R_{1}^{-\eta}+R_{2}^{-\eta})}-
\frac{2\pi T}{(n-2)(R_{1}^{-\eta}+R_{2}^{-\eta})} \bigg\{\lambda_{1}R_{2}^{2-\eta}\mathstrut_2 F_1\left(1,1-\frac{2}{\eta},2-\frac{2}{\eta},-\frac{TR_{2}^{-\eta}}{R_{1}^{-\eta}+R_{2}^{-\eta}}\right)+
\frac{\lambda_{2}P_{2}r_{1}^{2-\eta}}{P_{1}}\cdot \\ \mathstrut_2 F_1\left(1,1-\frac{2}{\eta},2-\frac{2}{\eta},\frac{-P_{2}Tr_{1}^{-\eta}}{P_{1}(R_{1}^{-\eta}+R_{2}^{-\eta})}\right)\bigg\}\Bigg)dr_{1}dR_{2}dR_{1}.
\label{cp3}
\end{multline}
\hrulefill
\begin{multline}
\mathcal{C}_{m,m(bk,IC)}^{(FD)}=\int_{0}^{\infty}\int_{R_{1}}^{\infty}\int_{0}^{R_{1}(\frac{P_{2}}{P_{1}})^{1/\eta}}\frac{(2\pi)^{3}}{A_{bk}^{(FD)}}{\lambda_{1}}^{2}\lambda_{2}R_{1} R_{2} r_{1} \exp\Bigg(-\pi(\lambda_{1}R_{2}^{2}+\lambda_{2}r_{1}^{2})-
\frac{2\pi T}{(n-2)(R_{1}^{-\eta}+R_{2}^{-\eta})}\cdot \\ \bigg\{\lambda_{1}R_{2}^{2-\eta}\mathstrut_2 F_1\left(1,1-\frac{2}{\eta},2-\frac{2}{\eta},-\frac{TR_{2}^{-\eta}}{R_{1}^{-\eta}+R_{2}^{-\eta}}\right)+
\frac{\lambda_{2}P_{2}r_{1}^{2-\eta}}{P_{1}}\mathstrut_2 F_1\left(1,1-\frac{2}{\eta},2-\frac{2}{\eta},\frac{-P_{2}Tr_{1}^{-\eta}}{P_{1}(R_{1}^{-\eta}+R_{2}^{-\eta})}\right)\bigg\}\Bigg)dr_{1}dR_{2}dR_{1}.
\label{cp3ic}
\end{multline}
\normalsize
\hrulefill
\end{figure*}
Considering two independent PPPs based two tier downlink cellular network with BS intensity $\lambda_{i}$ in a Rayleigh fading environment, the coverage probability for the blackout case is given in~\eqref{cp3}.
\end{theorem}
\begin{proof}
We obtain the coverage probability for the blackout user with cooperation by substituting the LTs found in Lemma~\ref{case3} in the conditional coverage probability expression given in~\eqref{cond3} and integrating it over the service distance distribution obtained in Lemma~\ref{dist3}.
\end{proof}
In the blackout phase, the user associates with the two strongest macro BSs while employing interference cancellation on the strongest femto BS. For the interference cancellation case, the blackout coverage probability is given by the following theorem.
\begin{theorem}
\label{theorem:Coverage Probability_IC}
Considering two independent PPPs based two tier downlink cellular network with BS intensity $\lambda_{i}$ in a Rayleigh fading environment, the coverage probability for blackout users in case \underline{FD} with IC capabilities is expressed in~\eqref{cp3ic}.
\end{theorem}
\begin{proof}
The theorem is obtained using the same methodology for obtaining Theorem~\ref{case3t} but with eliminating $I_{r_{1}}$ from \eqref{cond3} and integrating over the joint distance distribution expressed in~\eqref{jointdist3}.
\end{proof}
\subsection{\textbf{Macro Skipping Strategy {\rm (MS)}}}
In \underline{MS} scheme, the test user skips all femto BSs and every other macro BSs along its trajectory. Particularly, the user in this case alternates between the macro best connected and macro blackout modes. That is, in blackout phase, the test user skips the nearest macro BS and disregards the entire tier of femto BSs. Also, the cooperative non-coherent transmission from the second and the third nearest macro BSs is only activated during macro blackout. The conditional and joint service distance distributions for the test user in the blackout phase are given by the following lemma.
\begin{lem}
The joint distance distribution between the user and its skipped and serving or cooperating BSs in the blackout mode is given by
\small
\begin{align}
f_{R_{1},R_{2},R_{3}(bk)}^{(MS)}(x,y,z)=(2\pi\lambda_{1})^{3}xyz e^{-\pi\lambda_{1}z^{2}};\quad 0\leq x \leq y \leq z\leq \infty
\label{eq:joint}
\end{align}
\normalsize
The joint PDF of the distances between the test user and its serving or cooperating BSs in the blackout phase with BS cooperation is given by
\small
\begin{align}
f^{(MS)}_{R_2,R_3(bk)}(y,z)= 4(\pi\lambda)^3 y^{3}z e^{-\pi\lambda z^{2}};\quad 0 \leq y \leq z \leq \infty
\label{joint_m}
\end{align}
\normalsize
The conditional (i.e., conditioning on $R_{2}$) PDF of the distance between the test user and the skipped BS in the blackout case is given by
\small
\begin{eqnarray}
f_{R(bk)}^{(MS)}(R_{1}|R_{2})=\frac{2R_{1}}{R_{2}^{2}}.
\label{cond41}
\end{eqnarray}
\normalsize
The joint distance distribution between the user and the disregarded femto and serving macro BSs in the non-blackout mode is given by
\small
\begin{eqnarray}
f_{R_{1},r_{1}(\bar{bk})}^{(MS)}(x,y)=\frac{(2\pi)^{2}}{A_{f}}\lambda_{1}\lambda_{2}x_{1}y_{1}\exp(-\pi\lambda_{1}x^{2}-\pi\lambda_{2}y^{2}),
\label{joint4}
\end{eqnarray}
\normalsize
where $A_{f}$ is the probability that $P_{2}r_{1}^{-\eta}>P_{1}R_{1}^{-\eta}$, which is same as $A_{f}^{(BC)}$, given in~\eqref{Af}.
The marginal distribution of the distance between the user and its serving macro BS in non-blackout mode is given by
\small
\begin{eqnarray}
f_{R_1(\bar{bk})}^{(MS)}(x)=\frac{2\pi}{A_{f}}\lambda_{1}x\left(\exp\left(-\pi\lambda_{1}x^{2}\right)-\exp\left(-\pi x^{2}\left(\lambda_{1}+\lambda_{2}\left(\frac{P_{2}}{P_{1}}\right)^{2/\eta}\right)\right)\right).
\label{4m}
\end{eqnarray}
\normalsize
\label{dist4}
\end{lem}
\begin{proof}
The joint conditional distribution of $R_1$ and $R_2$ is the order statistics of two $i.i.d.$ random variables with PDF $\frac{2R}{R_{3}^{2}}$, where $0 \leq R \leq R_3$. The joint conditional distribution is given by $f_{R_1,R_2}(x,y \vert R_3) = \frac{8xy}{R_{3}^{4}}$, where $0<x<y<R_3$. By following Bayes' theorem, the joint PDF $f^{(MS)}_{R_1,R_2,R_3(bk)}(.,.,.)$ is obtained by multiplying the conditional joint PDF of $R_1$ and $R_2$ by the marginal PDF of $R_3$. The lemma follows by performing this marginalization over $R_3$, using its marginal distribution derived in eq. (2) in~\cite{dist}. The joint PDF of $R_2$ and $R_3$ is obtained by integrating \eqref{eq:joint} w.r.t. $x$ from 0 to $y$. The conditional PDF is obtained by dividing the joint PDF in~\eqref{eq:joint} by the marginal distribution in~\eqref{joint_m}. The joint PDF $f_{R_{1},r_{1}(\bar{bk})}^{(MS)}(.,.)$ is obtained by using the null probability of independent PPPs and the marginal distribution in~\eqref{4m} is found by integrating~\eqref{joint4} w.r.t. $y$ from 0 to $x(\frac{P_{2}}{P_{1}})^{1/\eta}$.
\end{proof}
\subsubsection{Coverage Probability}
By the law of total probability, we can write the overall coverage probability as
\small
\begin{align}
\mathcal{C}^{(MS)}=A_{m(\bar{bk})}^{(MS)}\mathcal{C}_{m(\bar{bk})}^{(MS)}+A_{f(\bar{bk})}^{(MS)}\mathcal{C}_{m\prime(\bar{bk})}^{(MS)}+
A_{bk}^{(MS)}\mathcal{C}_{m,m(bk)}^{(MS)},
\end{align}
\normalsize
where $\mathcal{C}_{m(\bar{bk})}^{(MS)}$ is the coverage probability for the best connected macro user while $\mathcal{C}_{m\prime(\bar{bk})}^{(MS)}$ is the coverage probability for the macro user when the strongest femto candidate is disregarded (i.e., $P_{2}r_{1}^{-\eta}>P_{1}R_{1}^{-\eta}$). Also, $\mathcal{C}_{m,m(bk)}^{(MS)}$ is the blackout coverage probability as the user skips the nearest macro BS and associates with the second and the third strongest macro BSs. Since the test user skips every other macro BS, the user spends $50\%$ of the time in association with the strongest macro and rest of the time in the blackout phase on average. Consequently, we assume $A_{bk}^{(MS)}$ to be $0.5$. Moreover, in the non-blackout case, the user is either in best connected mode (i.e. $P_{1}R_{1}^{-\eta}>P_{2}r_{1}^{-\eta}$) or it disregards the strongest femto and associates with the macro BS (i.e. $P_{2}r_{1}^{-\eta}>P_{1}R_{1}^{-\eta}$). Thus, $A_{m(\bar{bk})}^{(MS)}$ is considered to be $0.5*(1-A_{f(\bar{bk})})$, where $A_{f(\bar{bk})}=A_{f}^{(BC)}$, which is defined in~\eqref{Af}. Also, $\mathcal{C}_{m(\bar{bk})}^{(MS)}$ is the same as the best connected coverage probability for macro association as expressed in case \underline{BC} (i.e., $\mathcal{C}_{m}^{(BC)}$). However, $\mathcal{C}_{m\prime(\bar{bk})}^{(MS)}$ can be written as
\small
\begin{align}
\mathcal{C}_{m\prime(\bar{bk})}^{(MS)}=\mathbb{P}\bigg[\frac{P_{1} h_{1} R_{1}^{-\eta}}{I_{R}+I_{r_1}+I_{r}+\sigma^{2}}>T\bigg],
\end{align}
\normalsize
where $I_{R}$ is the aggregate interference power from the entire macro tier except $\{b_{1}\}$ while $I_{r_{1}}$ is the interference power from the strongest femto BS which lies from 0 to $R_{1}(\frac{P_{2}}{P_{1}})^{1/\eta}$ and $I_{r}$ is the aggregate interference power received from all femtos except $\{b_{1}\}$.
The blackout coverage probability case can be expressed as
\small
\begin{align}
\mathcal{C}_{m,m(bk)}^{(MS)}= \mathbb{P}\bigg[\frac{|\sqrt{P_{1}}h_{2} R_{2}^{-\eta/2}+\sqrt{P_{1}}h_{3} R_{3}^{-\eta/2}|^2}{I_{R_{1}}+I_{R}+I_{r}+\sigma^2}>T\bigg],
\end{align}
\normalsize
where $I_{R_1}$ is the received interference power from the nearest skipped macro BS while $I_{R}$ is the aggregate interference power from the whole macro tier except $\{b_{1},b_{2},b_{3}\}$. Here, $I_r$ is the aggregate interference power received from the whole femto tier, which exists from 0 to $\infty$. We define $I_{R_{1}}$ and $I_{R}$ for the blackout case as
\small
\begin{align*}
I_{R_{1}}= P_{1}h_{1}R_{1}^{-\eta}\hspace{0.2cm},\quad I_{R}= \sum_{i\epsilon\phi_1\backslash b_{1},b_{2},b_{3}}^{} P_{1}h_{i}R_{i}^{-\eta}.
\end{align*}
\normalsize
Since, $h_{i}\sim\exp(1)$, the conditional coverage probability for the non-blackout is given by
\small
\begin{align}
\mathcal{C}_{m\prime(\bar{bk})}^{(MS)}(R_{1})=\exp\bigg(\frac{-T}{P_{1}R_{1}^{-\eta}}\bigg)\mathscr{L}_{I_R}\hspace{-0.05cm}\bigg(\hspace{-0.05cm}\frac{T}{P_{1}R_{1}^{-\eta}}\bigg)\mathscr{L}_{I_{r_{1}}}\bigg(\hspace{-0.05cm}\frac{T}{P_{1}R_{1}^{-\eta}}\hspace{-0.05cm}\bigg)\mathscr{L}_{I_{r}}\bigg(\hspace{-0.05cm}\frac{T}{P_{1}R_{1}^{-\eta}}\hspace{-0.05cm}\bigg),
\label{cond42}
\end{align}
\normalsize
where the LTs of $I_{r_1}$ and $I_{r}$ are the same as given in~\eqref{Ir1} and~\eqref{Ir}, respectively, with the only difference that there is no $R_{2}$ in this case as the user is connected to one macro BS only. Also, $\mathscr{L}_{I_R}(s)$ is given in~\eqref{LT1}.
Using the conditional coverage probability expression and the service distance distribution for the best connectivity associations (i.e., non-blackout), the following theorem is obtained for the coverage probability.
\begin{theorem}
The coverage probability for macro association while disregarding the nearest femto BS in the non-blackout case is given by~\eqref{cp41}.
\end{theorem}
\begin{proof}
The theorem is proved by substituting the LTs in the conditional coverage probability expression~\eqref{cond42} and integrating over the service distance distribution found in~\eqref{4m}.
\end{proof}
Since, $h_{i}$'s are $i.i.d.$ $\mathcal{CN}$(0,1) such that $|x_{2}h_2+x_{3}h_3|^2\sim \exp\left(\frac{1}{x_{2}^{2}+x_{3}^{2}}\right)$, we can write the conditional coverage probability for the blackout user as
\small
\begin{align}
\mathcal{C}_{m,m(bk)}^{(MS)}(R_{2},R_{3})=exp\bigg(\frac{-T\sigma^2}{x_{2}^{2}+x_{3}^{2}}\bigg)\mathscr{L}_{I_{R_{1}}}\bigg(\frac{T}{x_{2}^{2}+x_{3}^{2}}\bigg)
\mathscr{L}_{I_{R}}\bigg(\frac{T}{x_{2}^{2}+x_{3}^{2}}\bigg)\mathscr{L}_{I_{r}}\bigg(\frac{T}{x_{2}^{2}+x_{3}^{2}}\bigg),
\label{cond4}
\end{align}
\normalsize
where $x_{i}$ is the same as defined in case \underline{FD}. Also, note that the LT of $I_r$ in the blackout case is different from the one expressed in case \underline{FD}. Here, we consider that the femto BSs can exist anywhere from 0 to $\infty$.
The LTs of the $I_{R_{1}}$, $I_{R}$ and $I_{r}$ for the blackout mode are expressed in the lemma below.
\begin{lem}
The LT of $I_{R_{1}}$ in the blackout mode with cooperative service from the second and third strongest macro BSs is given by
\small
\begin{align}
\mathscr{L}_{I_{R_{1}}}(s)=\int_{0}^{R_{2}}\frac{2R_{1}}{R_{2}^{2}(1+sP_{1}R_{1}^{-\eta})}dR_{1}.
\end{align}
\normalsize
The LT of $I_R$ in the blackout mode with BS cooperation can be expressed in terms of the hypergeometric function as
\small
\begin{align}
\mathscr{L}_{I_{R}}(s)=\exp\left(\frac{-\pi\lambda_{1}sP_{1}R_{3}^{2-\eta}}{\eta-2}\mathstrut_2 F_1\left(1,1-\frac{2}{\eta},2-\frac{2}{\eta},\frac{-sP_{1}}{R_{3}^{\eta}}\right)\right).
\end{align}
\normalsize
The LT of $I_r$ in the blackout case is given by
\small
\begin{align}
\mathscr{L}_{I_{r}}(s)=\exp\left(-2\pi^2\lambda_{2}\frac{(sP_{2})^{2/\eta}}{\eta}\csc\left(\frac{2\pi}{\eta}\right)\right).
\end{align}
\normalsize
\label{lt4}
\end{lem}
\begin{proof}
The LT of $I_{R_1}$ is obtained using the same procedure as done for $I_{r_1}$ in~\eqref{Ir1} but considering $s=\frac{T}{P_{1}(R_{2}^{-\eta}+R_{3}^{-\eta})}$, interference region from 0 to $R_{2}$ and the conditional distribution obtained in~\eqref{cond41}. The LT of $I_{R}$ is obtained in the similar way as $\mathscr{L}_{I_R}(s)$ for case \underline{FD} but with $s$ mentioned above and the interference boundary from $R_{3}$ to $\infty$. For $\mathscr{L}_{I_r}(s)$, we follow the same procedure as of $\mathscr{L}_{I_R}(s)$ with femto interference limits from 0 to $\infty$.
\end{proof}
The above LTs evaluated at $\eta=4$ are boiled down to closed form expressions as given by the following corollary.
\begin{col}
For the special case ($\eta=4$), the LT of $I_{R_{1}}$ is given by
\small
\begin{align}
\mathscr{L}_{I_{R_{1}}}(s)\big|_{\eta=4}=1-\sqrt{\frac{T}{1+R_{2}^{4}R_{3}^{-4}}}\arctan{\left(\sqrt{\frac{1+R_{2}^{4}R_{3}^{-4}}{T}}\right)}.
\end{align}
\normalsize
The LT of $I_{R}$ evaluated at $\eta=4$ can be expressed as
\small
\begin{align}
\mathscr{L}_{I_{R}}(s)\big|_{\eta=4}=\exp\left(\hspace{-0.15cm}-\pi\lambda\sqrt{\hspace{-0.1cm}\frac{T}{R_{2}^{-4}\hspace{-0.1cm}+\hspace{-0.1cm}R_{3}^{-4}}}\arctan\left(\hspace{-0.05cm}\sqrt{\frac{TR_{2}^{4}}{R_{2}^{4}\hspace{-0.1cm}+\hspace{-0.1cm}R_{3}^{4}}}\right)\right).
\end{align}
\normalsize
The LT of $I_{r}$ at $\eta=4$ is given by
\small
\begin{align}
\mathscr{L}_{I_{r}}(s)\big|_{\eta=4}=\exp\left(-\frac{\pi^{2}\lambda_{2}}{2}\sqrt{\frac{TP_{2}}{P_{1}(R_{2}^{-4}+R_{3}^{-4})}}\right).
\end{align}
\normalsize
\end{col}
Using the service distance distribution and the LTs in Lemmas~\ref{dist4} and~\ref{lt4}, we obtain the coverage probability for the case \underline{MS} as shown in the following theorem.
\begin{theorem}
\label{case4t}
\begin{figure*}
\scriptsize
\begin{multline}
\mathcal{C}_{m\prime(\bar{bk})}^{(MS)}=\int_{0}^{\infty}\int_{0}^{R_{1}(\frac{P_{2}}{P_{1}})^{1/\eta}}\frac{4\pi^{2}\lambda_{1}\lambda_{2}r_{1}R_{1}\exp(-\pi\lambda_{2}r_{1}^{2})\big(\exp(-\pi\lambda_{1}R_{1}^{2})-\exp(-\pi R_{1}^{2}(\lambda_{1}+\lambda_{2}(\frac{P_{2}}{P_{1}})^{2/\eta}))\big)}{\bigg(1+\frac{TP_{2}r_{1}^{-\eta}}{P_{1}R_{1}^{-\eta}}\bigg)A_{f(\bar{bk})}^{(MS)}\bigg(1-e^{-\lambda_{2}\pi R_{1}^{2}(\frac{P_{2}}{P_{1}})^{2/\eta}}\bigg)}\cdot\\
\exp\bigg(\frac{-2\pi T}{(\eta-2)}\bigg\{\frac{\lambda_{2}P_{2}T}{P_{1}R_{1}^{-\eta}}r_{1}^{2-\eta}\mathstrut_2 F_1\bigg(1,1-\frac{2}{\eta},2-\frac{2}{\eta},\frac{-P_{2}Tr_{1}^{-\eta}}{P_{1}R_{1}^{-\eta}}\bigg)+\lambda_{1}TR_{1}^{2}\mathstrut_2 F_1\bigg(1,1-\frac{2}{\eta},2-\frac{2}{\eta},-T\bigg)\bigg\}\bigg)dr_{1}dR_{1}.
\label{cp41}
\end{multline}
\hrulefill
\begin{multline}
\mathcal{C}_{m,m(bk)}^{(MS)}=\int_{0}^{\infty}\int_{R_2}^{\infty}4(\pi\lambda_1)^3 R_{2}^{3}R_{3}\int_{0}^{R_2}\frac{1}{1+\frac{TR_{1}^{-\eta}}{R_{2}^{-\eta}+R_{3}^{-\eta}}}\frac{2R_{1}}{R_{2}^{2}}dR_{1}\cdot
exp\bigg(-\pi\lambda_1 R_{3}^{2}-\frac{\pi\lambda_{1}TR_{3}^{2-\eta}}{(\eta-2)(R_{2}^{-\eta}+R_{3}^{-\eta})}\cdot\\ \mathstrut_2 F_1\left(1,1-\frac{2}{\eta},2-\frac{2}{\eta},-\frac{TR_{3}^{-\eta}}{R_{2}^{-\eta}+R_{3}^{-\eta}}\right)-\frac{2\pi^2\lambda_{2}}{\eta}\bigg(\frac{TP_{2}}{R_{2}^{-\eta}+R_{3}^{-\eta}}\bigg)^{2/\eta}\csc\bigg(\frac{2\pi}{\eta}\bigg)\bigg)dR_{3}dR_{2}.
\label{cp4}
\end{multline}
\hrulefill
\normalsize
\end{figure*}
Considering two independent PPPs based two tier downlink cellular network with BS intensity $\lambda_{i}$ in a Rayleigh fading environment, the coverage probability for the blackout user in case \underline{MS} with BS cooperation is given in~\eqref{cp4}.
\end{theorem}
\begin{proof}
We obtain the coverage probability for the blackout user with cooperation by substituting the LTs found in Lemma~\ref{lt4} in the conditional coverage probability expression given in~\eqref{cond4} and integrating it over the service distance distribution obtained in Lemma~\ref{dist4}.
\end{proof}
The coverage probability for the blackout user with interference cancellation capabilities is given by the following theorem.
\begin{theorem}
\label{theorem:Coverage Probability_IC}
Considering two independent PPPs based two tier downlink cellular network with BS intensity $\lambda_{i}$ in a Rayleigh fading environment, the coverage probability for blackout users in the case \underline{MS} with interference cancellation capabilities can be expressed as
\small
\begin{multline}
\mathcal{C}_{m,m(bk,IC)}^{(MS)}=\int_{0}^{\infty}\int_{R_2}^{\infty}4(\pi\lambda_1)^3 R_{2}^{3}R_{3}\exp\bigg(-\pi\lambda_1 R_{3}^{2}- \mathstrut_2 F_1\left(1,1-\frac{2}{\eta},2-\frac{2}{\eta},-\frac{TR_{3}^{-\eta}}{R_{2}^{-\eta}+R_{3}^{-\eta}}\right)\cdot\\
\frac{\pi\lambda_{1}TR_{3}^{2-\eta}}{(\eta-2)(R_{2}^{-\eta}+R_{3}^{-\eta})}- \frac{2\pi^2\lambda_{2}}{\eta}\bigg(\frac{TP_{2}}{R_{2}^{-\eta}+R_{3}^{-\eta}}\bigg)^{2/\eta}\csc\bigg(\frac{2\pi}{\eta}\bigg)\bigg)dR_{3}dR_{2}.
\end{multline}
\normalsize
\end{theorem}
\begin{proof}
The theorem is obtained using the same methodology for obtaining Theorem~\ref{case4t} but with eliminating $I_{R_{1}}$ from \eqref{cond4}.
\end{proof}
Figs.~\ref{cplot} and~\ref{cplot1} show the analysis and simulations for the coverage probabilities for all of the considered HO schemes without and with nearest BS interference cancellation. While the analysis is for stationary PPPs, the simulations in Figs.~\ref{cplot} and~\ref{cplot1} account for user mobility. Consequently, the good match between the analysis and simulations validates our model. Fig.~\ref{cplot} shows the cost of HO skipping from the coverage probability perspective. That is, sacrificing the best SINR connectivity degrades the coverage probabilities even with BS cooperation. Such coverage probability degradation can be mitigated via IC as shown in Fig.~\ref{cplot1}. For instance, employing BS cooperation and IC, the coverage probability for the \underline{FS} scheme is almost similar to the \underline{BS} scheme. Although the proposed HO schemes degrade the coverage probability, they offer tangible improvements to the average throughput due to decrease in the the HO rate as shown in the next sections.
\begin{figure*}[t!]
    \centering
    \begin{subfigure}[t]{0.5\textwidth}
        \centerline{\includegraphics[width= 0.90 \linewidth]{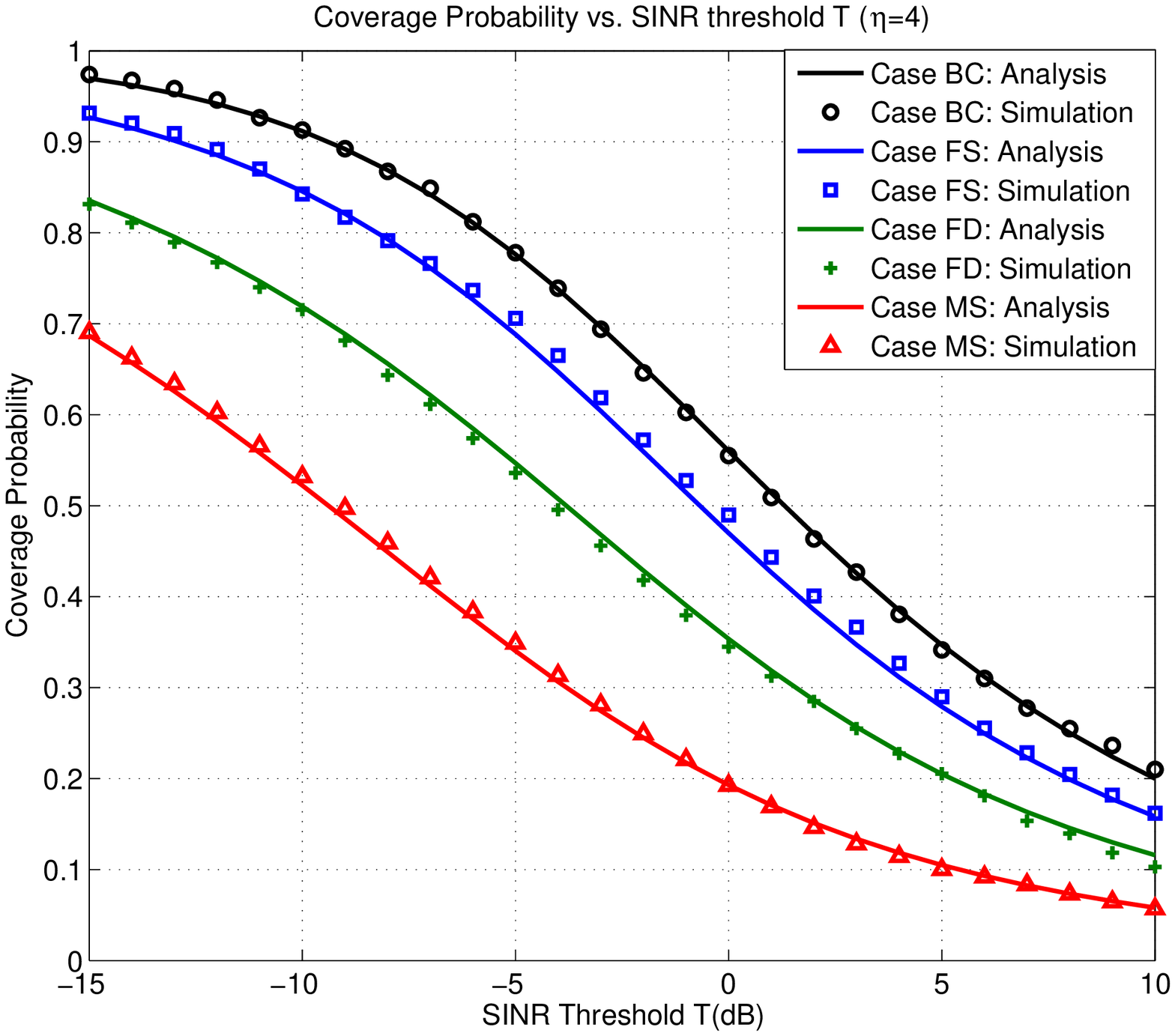}}\caption{Without Interference Cancellation}
\label{cplot}
    \end{subfigure}%
    ~
    \begin{subfigure}[t]{0.5\textwidth}
       \centerline{\includegraphics[width= 0.90 \linewidth]{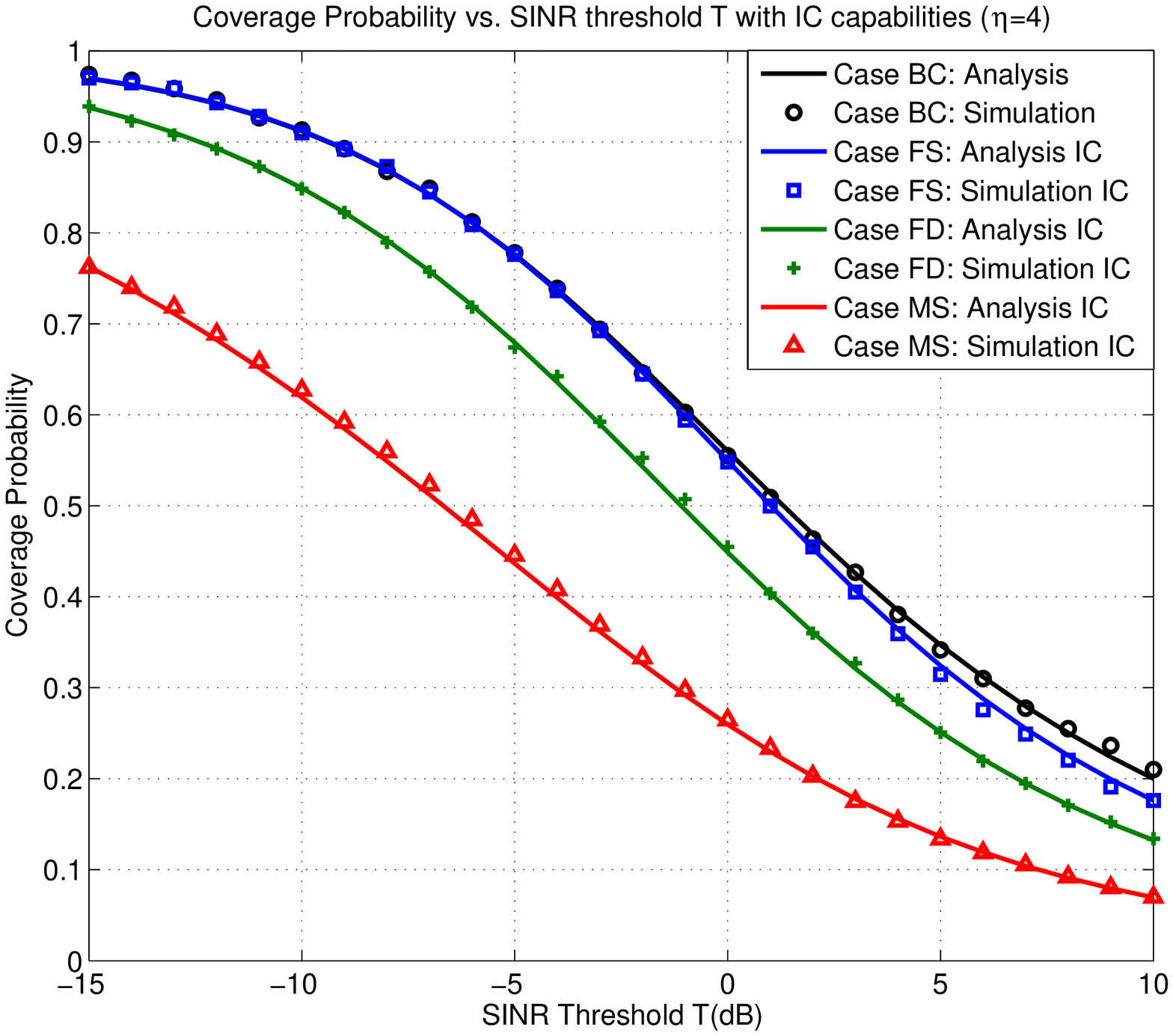}}\caption{With Interference Cancellation}
\label{cplot1}
    \end{subfigure}
  \small  \caption{Coverage probability plots for all cases at $\eta=4$, transmission power $P_{1}=1$ watt, $P_{2}=0.1P_{1}$ watt and BS intensities $\lambda_{1}=30$ $BS/Km^{2}$, $\lambda_{2}=70$ $BS/Km^{2}$}
\end{figure*}
\section{Handover Cost} \label{sec4}
In this section, we encompass user mobility effect and compute handover rates and cost for each HO skipping scheme. We define HO cost $D_{HO}$ as the normalized average time wasted during the execution of HO per unit time. Thus, $D_{HO}$ is the fraction of time where no data (i.e., control only) is transmitted to the test user. Note that the HO cost is different for each HO scheme due to different employed skipping strategies. Let $d_{ij}$ be the delay per $i$ to $j$ handover and $H_{ij}$ be the number of HOs from tier $i$ to $j$ per unit time, then $D_{HO}$ can be expressed as
\small
\begin{eqnarray}
D_{HO}&=& \sum_{i}^{K}\sum_{j}^{K}H_{ij} * d_{ij}.
\end{eqnarray}
\normalsize
where $K$ is the number of network tiers, which is 2 in our case. Also, we use $d_{m}$ and $d_{f}$ to denote macro-to-macro HO delay and all femto related HO delays, respectively\footnote{We assume that $d_{m} \leq d_{f} $ because macro BSs usually have high speed dedicated (e.g., fiber-optic) connectivity to the core network. On the other hand, femto BSs may reach the core network via the macro BS through additional backhaul hop or via a shared ADSL/IP connectivity.}. The HO rate is defined as the number of intersections between the user's trajectory and the cell boundaries per unit time. Following \cite{10a}, the tier $i$ to tier $j$ HO rate is given by
\small
\[
H_{ij}=\begin{cases}
\frac{v}{\pi}L_{ij}\quad \text{ if }i\neq j,\\
\frac{2v}{\pi}L_{ij}\quad \text{if }i= j.
\end{cases}
\]
\normalsize
where $v$ is the user velocity and $L_{ij}$ denotes the number of voronoi cell boundaries between a tier $i$ and tier $j$ BSs per unit length, which is given by
\small
\[
L_{ij}=\begin{cases}
\frac{\lambda_{i}\lambda_{j}F(x_{ij})}{2(\sum_{n=1}^{K}\lambda_{n}x_{nk}^{2})^\frac{3}{2}}+\frac{\lambda_{i}\lambda_{j}F(x_{ji})}{2(\sum_{n=1}^{K}\lambda_{n}x_{nj}^{2})^\frac{3}{2}} \quad \text{ if }i\neq j,\\
\frac{\lambda_{i}^{2}F(1)}{2(\sum_{n=1}^{K}\lambda_{n}x_{nk}^{2})^\frac{3}{2}} \quad \quad \quad \quad \quad \quad \quad \quad \quad \text{if }i= j.
\end{cases}
\]
\normalsize
where $x_{11}=x_{22}=1$, $x_{12}=\big(\frac{P_{1}}{P_{2}}\big)^{1/\eta}$, $x_{21}=\frac{1}{x_
{x_{12}}}$\quad
\small
\begin{eqnarray}
F(x)=\frac{1}{x^{2}}\int_{0}^{\pi}\sqrt{(x^{2}+1)-2 x cos(\theta)}d\theta.
\end{eqnarray}
\normalsize
In the \underline{BC} scheme, the user experiences all types of HOs i.e. horizontal and vertical HOs. Thus, the total HO cost in case \underline{BC} is given by
\small
\begin{eqnarray}
D_{HO}^{(BC)}=H_{11}d_{m}+(H_{12}+H_{21}+H_{22})d_{f}.
\end{eqnarray}
\normalsize
In \underline{FS} scheme, the user skips every other femto BS and associates to all macro BSs. Therefore, the HO rate from femto-to-femto and from macro-to-femto is reduced to half. Thus, we can write $D_{HO}$ for case \underline{FS} as
\small
\begin{eqnarray}
D_{HO}^{(FS)}=H_{11}d_{m}+\frac{H_{12}+H_{21}+H_{22}}{2} {d_{f}}.
\end{eqnarray}
\normalsize
The user in the \underline{FD} scheme skips all the femto BSs and associate to all macro BSs. Thus, $D_{HO}$ can be written as
\small
\begin{eqnarray}
D_{HO}^{(FD)}=H_{11} d_{m}.
\end{eqnarray}
\normalsize
In case \underline{MS}, the user disregards all femto BSs and skips every other macro BS. That is, the user spends $50\%$ time in macro best connected phase and rest of the $50\%$ in the macro blackout phase. Hence, we can write $D_{HO}$ as
\small
\begin{eqnarray}
D_{HO}^{(MS)}=\frac{H_{11}}{2}d_{m}.
\end{eqnarray}
\normalsize
\begin{figure}[!t]
\centering
\includegraphics[width=0.5 \linewidth]{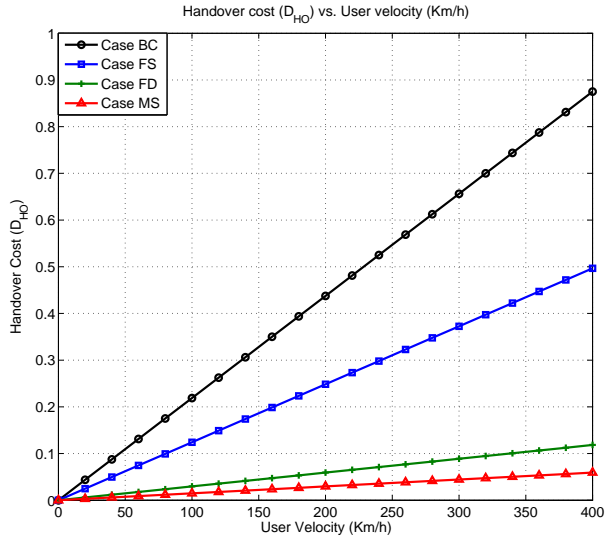}
\small \caption{Handover cost vs. User velocity (Km/h) with $P_{1}=1$ watt, $P_{2}=0.1P_{1}$ watt, $\lambda_{1}=30$ $BS/Km^{2}$, $\lambda_{2}=70$ $BS/Km^{2}$, $d_{m}=0.35$ s, $d_{f}=2d_{m}$ s }
\label{dho}
\end{figure}
Fig.~\ref{dho} shows the HO cost for each HO skipping strategy. It can be observed that the HO cost increases with the increase in user velocity.
\section{User Throughput}
In this section, we derive an expression for the user throughput, which is applicable to all HO skipping cases. In order to calculate the throughput, we need to omit the HO execution period. Thus, the average throughput (AT) can be expressed as
\small
\begin{eqnarray}
AT &=& W\mathcal{R}(1-D_{HO}).
\end{eqnarray}
\normalsize
where $W$ is the overall bandwidth of the channel and $\mathcal{R}$ is the achievable rate per unit bandwidth (i.e., nats/sec/Hz), which can be expressed as
\small
\begin{eqnarray}
\mathcal{R}=\ln(1+\theta)\mathbb{P}[SINR>\theta].
\end{eqnarray}
\normalsize
 By performing the numerical evaluation for achievable rate per unit bandwidth in each case, we get $\mathcal{R}$ in nats/sec/Hz as shown in table \ref{tab}.
\begin{table}[ht]
\renewcommand{\arraystretch}{1.3}
\caption{Achievable rate for all cases in nats/sec/Hz (T=6 dB)}
\label{tab}
\centering
 \begin{tabular}{||c c c||}
 \hline
 \multicolumn{3}{||c||}{\textbf{Achievable rate (nats/sec/Hz)}} \\
 \hline
 Case & Non-IC & IC \\
 \hline\hline
 Best Connected $\mathcal{R}^{(BC)}$ & 0.50 & - \\
 \hline
 Femto Skipping $\mathcal{R}^{(FS)}$ & 0.40 & 0.46 \\
 \hline
 Femto Disregard $\mathcal{R}^{(FD)}$ & 0.29 & 0.36 \\
 \hline
 Macro Skipping $\mathcal{R}^{(MS)}$ & 0.15 & 0.20\\
 \hline
\end{tabular}
\end{table}

\subsection{Design Insights}

Now we study the user's average throughput for the proposed HO schemes and define the effective velocity regions for each of them for the network parameters shown in table \ref{tab2}. The results shown in Figs. \ref{throughput} and~\ref{throughput1} for the average throughput account for IC and consider various values for the macro and femto HO delays.

Figs. \ref{throughput} and~\ref{throughput1} emphasize the HO problem in dense cellular environments in which the legacy best connected HO strategy imposes severe degradation to the average throughput as the user velocity increases. The figures also show that each of the proposed HO skipping strategies provides an effective solution for the HO problem in a certain velocity range. For instance, once the user velocity exceeds 100 Km/h, the femto skipping (\underline{FS}) strategy provides more than $10\%$ increase in the average throughput as compared to the best connected association (cf. Fig.~\ref{throughput1}). Furthermore, the proposed adaptive HO skipping results show up to $77\%$ gains in the average throughput as compared to the best connected association for the user velocity ranging from 80 Km/h to 200 Km/h. However, it is worth noting that the cases \underline{FS} and \underline{FD} show gains in the average throughput at medium and high velocity ranges, respectively. Also, we can observe that the skipping of macros in a two tier network outperforms the RSS based association at very high user velocities.

\begin{table}[ht]
\centering
\renewcommand{\arraystretch}{1.3}
\caption{Simulation parameters}
\label{tab2}
 \begin{tabular}{||l|l||}
 \hline
 \multicolumn{2}{||c||}{\textbf{Simulation parameters}} \\
 \hline
\hline  Overall Bandwidth $W$ : \quad 10 MHz &Path loss exponent $\eta$:\quad 4\\
\hline  SINR Threshold $\theta$ : \quad 6 dB &HO delay $d_{m,f}$:\quad 0.35, 0.7, 1.05 s\\
\hline  Macros intensity $\lambda_{1}$:  30 $BS/Km^{2}$ &Femtos intensity $\lambda_{2}$: 70 $BS/Km^{2}$\\
 \hline Macro Tx Power $P_{1}$: \quad 1 watt &Femto Tx Power $P_{2}$: \quad 0.1 watt\\
 \hline
          \end{tabular}
        \end{table}
\begin{figure*}[t!]
    \centering
    \begin{subfigure}[t]{0.5\textwidth}
        \centerline{\includegraphics[width= 0.90 \linewidth]{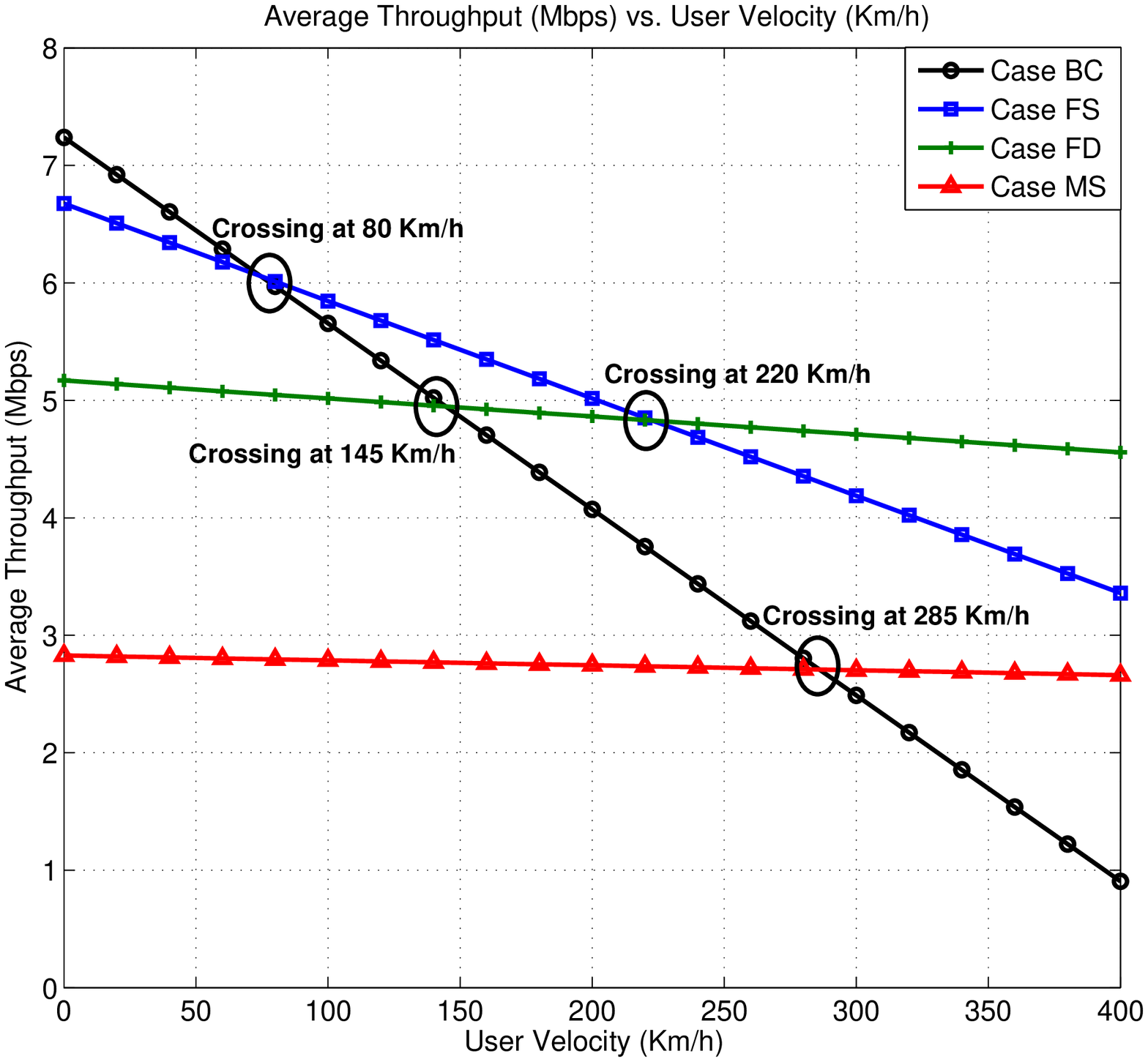}}\caption{$d_{m} = 0.35, d_{f} = 2d_{m}$  second}
\label{throughput}
    \end{subfigure}%
    ~
    \begin{subfigure}[t]{0.5\textwidth}
       \centerline{\includegraphics[width=  0.90 \linewidth]{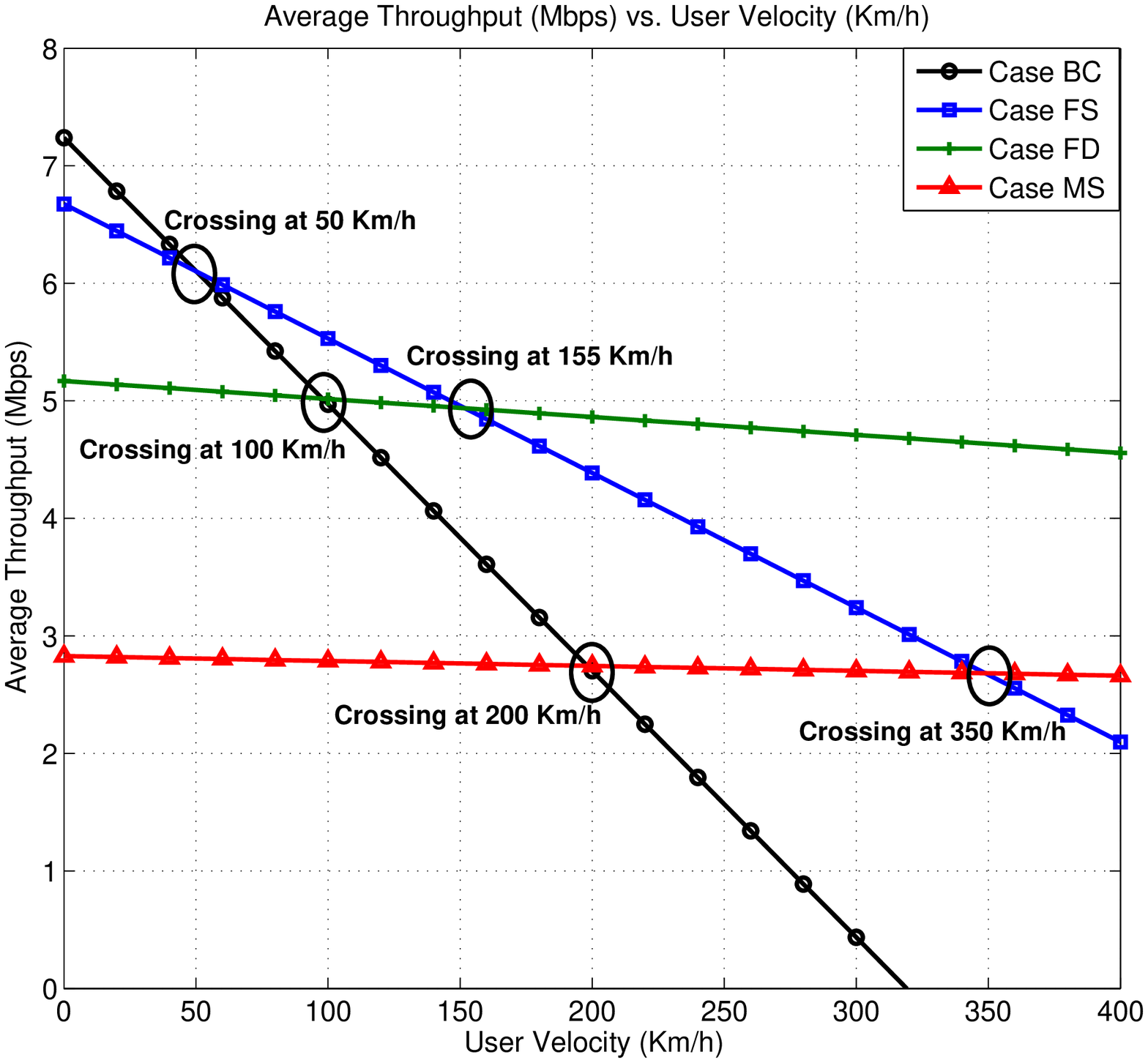}}\caption{$d_{m} = 0.35, d_{f} = 3d_{m}$ second}
\label{throughput1}
    \end{subfigure}
  \small  \caption{Average Throughput (Mbps) vs. User velocity (Km/h)}
\end{figure*}
\section{Conclusion}
This paper proposes user velocity aware HO skipping schemes for two tier cellular network to enhance the average rate for mobile users. We develop an analytical paradigm to model the performance of the proposed cooperative HO skipping schemes in order to study the effect of HO delay on the user rate. The developed mathematical model is based on stochastic geometry and is validated via Monte Carlo simulations. The results manifest the negative impact of HO on the users' rate in dense cellular networks and emphasize the potential of the proposed HO schemes to mitigate such negative HO impact. Particularly, the results show up to $77\%$ more rate gains, which can be harvested via the proposed HO schemes when compared to the conventional HO scheme that always maintains the best RSS association.
For future work, we will extend our study towards location aware HO skipping. Thus, we will propose HO skipping based on user's trajectory, which will maximize the gains while meeting the quality of service requirements.
\appendices
\section{Proof of Lemma 2}
The Laplace transform of $I_R$ can be expressed as
\small
\begin{eqnarray}
\mathscr{L}_{I_R(m)}(s) &=& \mathrm{E}\lbrack e^{-sI_R}\rbrack\\
&=&\mathrm{E}\lbrack e^{-s\sum_{i\epsilon \phi_{1} \backslash b_1}{}P_{1}h_iR_{i}^{-\eta}}\rbrack.\notag
\end{eqnarray}
\normalsize
Due to the independence between fading coefficients and BSs locations, we get
\small
\begin{eqnarray}
\mathscr{L}_{I_R(m)}(s)&=&\mathrm{E}_{\phi} \bigg\{ \prod_{i\epsilon \phi_{1} \backslash b_1}^{} \mathrm{E}_{h_i}\big\{e^{-sP_{1}h_{i} R_{i}^{-\eta}}\big\}\bigg\}\\
&=&\mathrm{E}_{\phi}\bigg\{\prod_{i\epsilon \phi_{1} \backslash b_1}^{} \mathscr{L}_{h_i}(sP_{1}R_{i}^{-\eta})\bigg\}.\notag
\end{eqnarray}
\normalsize
However, since $h_{i}$ $\sim$ $\exp(1)$, we can write
\small
\begin{eqnarray}
\mathscr{L}_{I_R(m)}(s)&=&\mathrm{E}_{\phi}\bigg\{\prod_{i\epsilon \phi_{1} \backslash b_1}^{}\frac{1}{1+sP_{1}R_{i}^{-\eta}}\bigg\}\notag.
\end{eqnarray}
\normalsize
Using probability generating functional (PGFL) for PPP \cite{20a}, we get
\small
\begin{align}
\mathscr{L}_{I_R(m)}(s)= \exp\bigg(-2\pi\lambda_{1}\int_{R_1}^{\infty}(1-\frac{1}{1+sP_{1}v^{-\eta}})vdv\bigg).
\end{align}
\normalsize
By change of variables $w=(sP_{1})^{-1/\eta}v$ and substituting $s=\frac{TR_{1}^{\eta}}{P_{1}}$, we have
\small
\begin{eqnarray*}
\mathscr{L}_{I_R(m)}(s)
&=& \exp\left(-2\pi \lambda_{1} R_{1}^{2} T^{2/\eta}\int_{T^{-1/\eta}}^{\infty}\frac{w}{1+w^{\eta}}dw\right)\notag\\ \notag\\
&=& \exp\bigg(-\frac{2\pi\lambda_{1}TR_{1}^{2}}{\eta-2}\mathstrut_2F_1\big(1,1-\frac{2}{\eta},2-\frac{2}{\eta},-T\big)\bigg).
\end{eqnarray*}
\normalsize
The LT of $I_{r}$ can be written as
\small
\begin{align}
\mathscr{L}_{I_r(m)}(s) =\mathrm{E}\bigg\{ e^{-s\sum_{i\epsilon \phi_{2}}{}P_{2}h_ir_{i}^{-\eta}}\bigg\}\notag.
\end{align}
\normalsize
Following the same procedure as shown for $\mathscr{L}_{I_R(m)}(s)$ above and considering the interference region of femto BSs from $R_{1}(\frac{P_{2}}{P{1}})^{1/\eta}$ to $\infty$, we get $\mathscr{L}_{I_r(m)}(s)$ as shown below
\small
\begin{align}
\mathscr{L}_{I_r(m)}(s)=\exp\bigg(\frac{-2\pi\lambda_{2}TR_{1}^{2}}{\eta-2}\big(\frac{P_{2}}{P_{1}}\big)^{2/\eta}\mathstrut_2F_1\big(1,1-\frac{2}{\eta},2-\frac{2}{\eta},-T\big)\bigg).
\end{align}
\normalsize
The LTs $\mathscr{L}_{I_{R}(f)}(s)$ and $\mathscr{L}_{I_{r}(f)}(s)$ in the femto association case are derived in the same manner but using the macro interference region from $r_{1}(\frac{P_{1}}{P{2}})^{1/\eta}$ to $\infty$ and femto interference from $r_{1} \rightarrow \infty$.
\section{Proof of Lemma 3}
First, we write intensity measure of the points inside a ball $\mathbf{B}$ of radius $r$ as $\Lambda(\mathbf{B})=\pi\lambda r^{2}$ and the intensity function, which is given by $\lambda(x)=2\pi\lambda r^{2}$. Then, using mapping theorem, we can write the intensity measure on a line from 0 to $y$ as $\Lambda([0,y])=\pi\lambda(Py)^{2/\eta}$ and the intensity function $\lambda(y)=\frac{2}{\eta}\pi\lambda P^{2/\eta}y^{2/\eta-1}$. Now, using superposition theorem, we can express the total intensity as
\small
\begin{equation}
\lambda(y)= \frac{2 \pi}{\eta} \left(\lambda_1 P_1^{2/\eta} + \lambda_2 P_2^{2/\eta}\right) y^{2/\eta-1}.
\end{equation}
\normalsize
\section{Proof of Lemma 4}
The conditional distance distribution of $r_1$ conditioned on the second strongest BS distance $x$ is given by
\small
\begin{eqnarray}
f_{r}(r_{1}|x) =\frac{\lambda_(r_1)}{\int_{0}^{x}\lambda(z)dz}\notag = \frac{2r_{1}^{2/\eta-1}}{\eta x^{2/\eta}}.
\end{eqnarray}
\normalsize
Using the null probability of PPP, we can find the service distance distribution in a single tier network as
\small
\begin{eqnarray*}
f_{Y}(y)=\frac{d}{dy}(1-e^{\pi\lambda(Py)^{2/\eta}}) = \frac{2}{\eta}\pi\lambda P^{2/\eta}y^{2/\eta-1}e^{-\pi \lambda P^{2/\eta}y^{2/\eta}}.
\end{eqnarray*}
\normalsize
Following the above PDF, we can write the PDF of $r_{1}$ (i.e., distance between the user and the strongest femto BS) in a two tier network as
\small
 \begin{eqnarray*}
f_{r_1}(r)=\frac{2\pi\lambda_{2}}{\eta A_{f}}P_{2}^{2/\eta}r^{2/\eta-1}e^{-\pi r^{2/\eta}(\lambda_{1}P_{1}^{2/\eta}+\lambda_{2}P_{2}^{2/\eta})}.
 \end{eqnarray*}
 \normalsize
 where $A_{f}$ is the probability that $r_{1}>R_{1}$ (i.e., femto BS provides the best SINR), which is the same as $A_{f}^{(BC)}$ in case \underline{BC}.
 Thus, we can write the distribution of $r_{1}$ as
 \small
 \begin{eqnarray}
 f_{r_1}(r)=\frac{2\pi\lambda_{t}}{\eta}\exp\big(-\pi r^{2/\eta}\lambda_{t}\big),
 \end{eqnarray}
 \normalsize
 where
 \small
\begin{eqnarray}
\lambda_{t}=\lambda_{1}P_{1}^{2/\eta}+\lambda_{2}P_{2}^{2/\eta}.
\end{eqnarray}
\normalsize
We can write the conditional distance distribution of the third strongest BS conditioning on $r_{1}$ as
\small
\begin{eqnarray}
P[x_{2}<y|r_{1}]=1-\exp\bigg(\int_{r_{1}}^{y}\frac{2\pi\lambda_{t}}{\eta}r^{2/\eta-1}dr\bigg)-\exp\bigg(\int_{r_{1}}^{y}\frac{2\pi\lambda_{t}r^{2/\eta-1}}{\eta}dr\bigg)\int_{r_{1}}^{y}\frac{2\pi\lambda_{t}r^{2/\eta-1}}{\eta/1!}dr. \end{eqnarray}
\normalsize
By differentiating above equation w.r.t. $y$, we get
\small
\begin{eqnarray*}
f(x_{2}|r_{1})=\frac{2}{\eta}(\pi\lambda_{t})^{2}y^{2/\eta-1}(y^{2/\eta}-r_{1}^{2/\eta})e^{-\pi\lambda_{t}(y^{2/\eta}-r_{1}^{2/\eta})}.
\end{eqnarray*}
\normalsize
$f_{x_1}(x)$ can be calculated as:
\small
\begin{eqnarray}
f_{x_1}(x)=\frac{\lambda(x)}{\int_{r_1}^{y}\lambda(z)dz}=\frac{2x^{2/\eta-1}}{\eta(y^{2/\eta-r_{1}^{2/\eta}})}.
\end{eqnarray}
\normalsize
Thus, we can write the joint conditional distribution as
\small
\begin{eqnarray*}
f_{x_1,x_2}(x,y|r_{1})=(\frac{2}{\eta}\pi\lambda_{t})^{2}(xy)^{2/\eta-1}\exp\big(-\pi\lambda_{t}(y^{2/\eta}-r_{1}^{2/\eta})\big).
\end{eqnarray*}
\normalsize
Now, we get the joint distribution $f_{x_1,x_2}(x,y,r_{1})$ as
\small
\begin{eqnarray*}
f_{x_1,x_2,r}(x,y,r_{1})=\bigg(\frac{2}{\eta}\pi\lambda_{t}\bigg)^{3}(x y r_{1})^{2/\eta-1}\exp(-\pi\lambda_{t}y^{2/\eta}).
\end{eqnarray*}
\normalsize
By integrating the above distribution w.r.t. $r_{1}$, from 0 $\rightarrow$ $x$, we get $f_{x_1,x_2}(x,y)$ as
\small
\begin{eqnarray}
\hspace{-0.2cm}f_{x_1,x_2}(x,y)=\frac{4}{\eta^{2}}\big(\pi\lambda_{t})^{3}x^{4/\eta-1}y^{2/\eta-1}\exp(-\pi \lambda_{t}y^{2/\eta}).
\end{eqnarray}
\normalsize
\section{Proof of Lemma 5}
The LT of $I_{r_1}$ can be expressed as
\small
\begin{eqnarray*}
\mathscr{L}_{I_{r_{1}}}(s) &=& \mathrm{E}\lbrack e^{-sI_{r_1}}\rbrack =\mathrm{E}\lbrack e^{-s\frac{h_{1}}{r_{1}}}\rbrack,
\end{eqnarray*}
\normalsize
Since, $h\sim\exp(1)$, we can write $\mathscr{L}_{I_{r_{1}}}(s)$ as
\small
\begin{eqnarray*}
\mathscr{L}_{I_{r_{1}}}(s) &=&\mathrm{E}\bigg\lbrack\frac{1}{1+s/r_{1}}\bigg\rbrack =\int_{0}^{x}\frac{1}{1+s/r_{1}}f(r_{1})dr_{1},
\end{eqnarray*}
\normalsize
Using~\eqref{20} obtained in Lemma~\ref{dist2} and substituting $s=\frac{T}{x^{-1}+y^{-1}}$, we can express $\mathscr{L}_{I_{r_{1}}}(s)$ as
\small
\begin{eqnarray}
\mathscr{L}_{I_{r_{1}}}(s)=\int_{0}^{x}\frac{2r_{1}^{2/\eta-1}}{\eta x^{2/\eta}\big(1+\frac{T}{r_{1}(x^{-1}+y^{-1})}\big)}dr_{1}.
\end{eqnarray}
\normalsize
Similarly, the LT of $I_{agg}$ can be written as
\small
\begin{eqnarray}
\mathscr{L}_{I_{agg}}(s)=\mathrm{E}\bigg\{ e^{-s\sum_{i\epsilon \phi \backslash b_1}{}h_i/u_{i}}\bigg\}.\notag
\end{eqnarray}
\normalsize
Due to the independence of the fading coefficients and the BSs locations and assuming $h_{i}$ $\sim$ $\exp(1)$, we get
\small
\begin{eqnarray}
\mathscr{L}_{I_{agg}}(s)&=&\mathrm{E}_{\phi}\bigg\{\prod_{i\epsilon \phi \backslash b_1}^{}\frac{1}{1+s/u_{i}^{-\eta}}\bigg\}.\notag
\end{eqnarray}
\normalsize
Applying PGFL for PPP, we get
\small
\begin{eqnarray*}
\mathscr{L}_{I_{agg}}(s)&=&\exp\bigg(-\frac{2\pi\lambda_{t}}{\eta}\int_{y}^{\infty}\frac{z^{2/\eta-1}}{1+z/s}dz\bigg).
\end{eqnarray*}
\normalsize
By substituting $s=\frac{T}{x^{-1}+y^{-1}}$ and simplifying the above equation, we get
\small
\begin{eqnarray*}
\mathscr{L}_{I_{agg}}(s)=\exp\bigg(\frac{-2\pi\lambda_{t}Ty^{2/n-1}}{(\eta-2)(x^{-1}+y^{-1})}\mathstrut_2 F_1\bigg(1,1-\frac{2}{\eta},2-\frac{2}{\eta},\frac{-T}{x^{-1}y+1}\bigg)\bigg).
\end{eqnarray*}
\normalsize
\ifCLASSOPTIONcaptionsoff
  \newpage
\fi

\bibliographystyle{IEEEtran}
\bibliography{IEEEabrv,mybibr}
\vfill

\end{document}